\documentclass[11pt,a4paper]{article}
\usepackage{fullpage}

\usepackage{amsmath,amssymb}
\usepackage{authblk}
\usepackage{color}
\usepackage{marginnote}
\usepackage[linesnumbered,vlined,ruled]{algorithm2e}
\usepackage{graphicx}
\usepackage{tikz}
 
\newtheorem{theorem}{Theorem}
\newtheorem{corollary}{Corollary}
\newtheorem{lemma}{Lemma}

\newtheorem{definition}{Definition}
\newtheorem{claim}{Claim}

\newcommand{\qed}{\hfill $\Box$ \medbreak}
\newenvironment{proof}{\noindent {\bf Proof.}}{\qed}

\newenvironment{proofofclaim}{\noindent {\emph{Proof.}}}{\hfill $\diamond$ \medbreak}

\newcommand{\dMAM}{\textsf{dMAM}}
 
\newcommand{\id}{\mbox{\rm id}}

\newcommand{\Forb}{\mathrm{Forb}}

\begin{document}

\title{Compact Distributed Certification of Planar Graphs}

\author[1]{Laurent Feuilloley\thanks{Additional support from MIPP and Jos\'e Correa's Amazon Research Award}}
\author[2]{Pierre Fraigniaud\thanks{Additional support from ANR project DESCARTES, and from INRIA project GANG}}
\author[3]{Ivan Rapaport\thanks{Additional support from CONICYT via PIA/Apoyo a Centros Cient\'{\i}ficos y Tecnol\'ogicos de Excelencia AFB 170001 and Fondecyt 1170021}}
\author[4]{\'Eric R\'emila\thanks{Additional support from IDEXLYON (project INDEPTH) within the Programme Investissements  d'Avenir (ANR-16-IDEX-0005) and  ``Math\'ematiques de la D\'ecision pour l'Ing\'enierie Physique et Sociale'' (MODMAD) }%
}
\author[5]{Pedro Montealegre\thanks{Additional support from CONICYT via  PAI + Convocatoria Nacional Subvenci\'on a la Incorporaci\'on en la Academia A\~no 2017 + PAI77170068 and FONDECYT  11190482}}
\author[6]{Ioan Todinca}

\affil[1]{Departamento de Ingenier\'{\i}a Industrial, Universidad de Chile, Chile}
\affil[2]{IRIF, CNRS and Universit\'e de Paris, France}
\affil[3]{DIM-CMM (UMI 2807 CNRS), Universidad de Chile, Chile}
\affil[4]{GATE Lyon St-Etienne  (UMR 5824  CNRS), UJM St-Etienne, France}
\affil[5]{Facultad de Ingenier\'{\i}a y Ciencias, Universidad Adolfo Iba\~nez, Santiago, Chile.}
\affil[6]{LIFO, Universit\'e d'Orl\'eans and INSA Centre-Val de Loire, France}

\date{}

\maketitle

\begin{abstract}
Naor, Parter, and Yogev (SODA 2020) have recently demonstrated the existence of a \emph{distributed interactive proof} for planarity (i.e., for certifying that a network is planar), using a sophisticated generic technique for constructing distributed IP protocols based on sequential IP protocols. The interactive proof for planarity  is based on a distributed certification of the correct execution of any given sequential linear-time algorithm for planarity testing. It involves three interactions between the prover and the randomized distributed verifier (i.e., it is a \dMAM\/ protocol), and uses small certificates, on $O(\log n)$ bits in $n$-node networks. We show that a single  interaction from the prover suffices, and randomization is unecessary, by providing an explicit description of a \emph{proof-labeling scheme} for planarity, still using certificates on just $O(\log n)$ bits. We also show that there are no proof-labeling schemes --- in fact, even no \emph{locally checkable proofs} --- for planarity using certificates on $o(\log n)$ bits. 
\end{abstract}

\vfill

\thispagestyle{empty}
\pagebreak
\setcounter{page}{1}
 
\section{Introduction}

Planar graphs are arguably among the most studied classes of graphs in the context of algorithm design and analysis. In particular, planar graphs enable the design of algorithms provably faster than for general graphs (see, e.g., \cite{Cabello19} for a subquadratic algorithm for diameter in planar graphs, in contrast to the quadratic lower bound in~\cite{RodittyW13}). In the context of distributed network computing too, planar graphs have been the source of many contributions. For instance, a distributed algorithm approximating a minimum dominating set (MDS) within a constant factor, in a constant number of rounds, has been designed for planar graphs~\cite{LenzenOW08}. In contrast, even a poly-logarithmic approximation of the MDS problem requires at least $\Omega(\sqrt{\log n/ \log\log n})$ rounds in arbitrary $n$-node networks~\cite{KuhnMW04}. See Section~\ref{subsec:related-work} for further references to contributions on distributed algorithm design for planar graphs. 

In this paper, we are concerned with \emph{checking} whether a given network is planar. Indeed, we aim at avoiding  the risk of performing distributed computations dedicated to planar graphs in networks that are not planar, which may lead to erroneous outputs, or even lack of termination. Note that it is sufficient that one node detects if the actual graph is not planar, as this node can then broadcast an alarm, or launch a recovery procedure (e.g., for overlay networks). 

Efficient planarity tests have been designed in the \emph{sequential} setting~\cite{HopcroftT74}, but no such efficient algorithms are known in the \emph{distributed} setting. In fact, it is known~\cite{FraigniaudKP13}  that there are no \emph{local} decision algorithms for planarity. That is, there are no algorithms in which every node of a network exchange information solely with nodes in its vicinity, and outputs accept or reject such that the network is planar if and only if all nodes accept. 

Kuratowski's theorem states that a graph is planar if and only if it does not contain a subgraph that is a subdivision of the complete graph $K_5$, or the complete bipartite graph $K_{3,3}$. This characterization can easily be turned into a \emph{proof-labeling scheme}~\cite{KormanKP10} that a network is \underline{not} planar, by certifying the presence of a subdivided $K_5$ or $K_{3,3}$ as a subgraph. However, it is not clear whether Kuratowski's theorem can be used to prove that a network \underline{is} planar, in a distributed manner. Indeed, this would require to certify that no subdivided $K_5$ or $K_{3,3}$ are present in the network. Also, it is not clear whether using the coordinates of the nodes in a planar embedding would help, as checking whether two edges cross seems difficult if the extremities of these edges are embedded at far away positions in the plane.  Similarly, a distributed proof based on checking the consistency of the faces resulting from a planar embedding of the graph appears uneasy, as one can construct embeddings of non planar graphs in which the faces look locally consistent. 

Yet, Kuratowski's theorem can be used to show that planarity is at the second level of the local decision \emph{hierarchy} with alternating quantifiers (see~\cite{FeuilloleyFH16}). That is, for every assignment of $O(\log n)$-bit certificates to the nodes by a \emph{disprover} aiming at convincing the nodes of a planar network that this network is not planar, a \emph{prover} can assign other $O(\log n)$-bit certificates to the nodes for certifying that the distributed proof provided by the disprover is erroneous. The fact that planarity is at a low level of the local decision hierarchy is conceptually informative, but this fact does not provide a concrete distributed mechanism for certifying planarity. 

A breakthrough has been recently achieved by Naor, Parter, and Yogev~\cite{NaorPY20}, who showed the existence of a \emph{distributed interactive proof} for planarity. Such types of protocols are motivated by the possible access to services provided by a computationally powerful external entity, e.g., the cloud~\cite{KolOS18}. The interactive protocol for planarity is not explicit, but can be constructed automatically from any sequential linear-time algorithm for planarity testing. In fact, this protocol is just one illustration of a sophisticated generic \emph{compiler} developed in~\cite{NaorPY20} for constructing distributed IP protocols based on sequential IP protocols, which applies way beyond the case of planarity. Specifically, given an algorithm for planarity testing, the interactive protocol is roughly the following. The $O(n)$ steps corresponding to the execution of the algorithm on the actual network are distributed to the nodes, and the interactive protocol checks that these steps are consistent (e.g., the output of step~$i$ is the input of step~$i+1$, etc.). The node handling the final step of the execution of the algorithm accepts or rejects according to the outcome of the algorithm. The steps of the execution (e.g., reading some memory location, or adding such and such registers) can be encoded on $O(\log n)$ bits, and it is shown that verifying the consistency of the $O(n)$ steps distributed over the $n$~nodes can be achieved by a \dMAM\/ interactive protocol,  using certificates on $O(\log n)$ bits. In such a protocol, the \emph{non-trustable} prover (a.k.a.~Merlin) provides each node with an $O(\log n)$-bit certificate. Then, at each node, the \emph{honest} verifier (a.k.a.~Arthur) running at that node challenges the prover with a random value. Finally the prover replies to each node by providing it with a second certificate, again on $O(\log n)$ bits. Once this is done, every node interacts with all  its neighbors only once, and outputs accept or reject. The authors of~\cite{NaorPY20} proved that, using their \dMAM\/ protocol, all nodes accept if and only if the network is planar. 

The mechanism in~\cite{NaorPY20} actually applies to all classes of graphs that can be recognized with a sequential polynomial-time algorithm. The resulting distributed  interactive protocol is however efficient only for classes of sparse graphs, recognizable in quasi-linear time, among which the class of planar graphs plays a prominent role.

\subsection{Our results}

We show that several interactions  between the prover and the verifier are not necessary for distributed certification of planarity, while the size of the certificates can be kept as small as in~\cite{NaorPY20}. Indeed, we provide an explicit description of a \emph{proof-labeling scheme}~\cite{KormanKP10}, still using certificates on just $O(\log n)$ bits. A proof-labeling scheme requires a single interaction between the prover and the distributed verifier, and no randomization is required. Specifically, the prover provides each node with an $O(\log n)$-bit certificate, and then the nodes can directly move on with the local verification stage, in which every node interacts with all  its neighbors once, and outputs accept or reject such that all nodes accept if and only if the network is planar. Proof-labeling schemes can be implemented even in absence of services provided by external entities like the cloud. Actually, in many frameworks, including the one in this paper, the certificates can be computed in a distributed manner by the network itself during a pre-processing phase. 

As mentioned before, it is not clear whether Kuratowski's theorem, or the use of coordinates can be applied to prove that a network is planar, in a distributed manner (Kuratowski's theorem can be used to certify that a network is \underline{not} planar). Therefore, we adopt a different approach, by asking the prover to certify a specific form of planar embedding of the network, not relying on coordinates. Our approach is inspired by the work on planar graphs by Fraysseix and Rosentiehl~\cite{FraysseixR85}, based itself on Tutte's crossing number theory~\cite{Tutte70}. Observe that any graph $G$ can be embedded in the plane as a planar embedding of a spanning tree $T$ of $G$, plus cotree edges (the  \emph{cotree} of $G$ associated with $T$ is the set of edges in $G$ not in $T$), such that the crossings occur  between cotree edges only. Such an embedding is called a \emph{$T$-embedding} of~$G$. A graph $G$ is planar if and only if there exists a $T$-embedding of $G$ in the plane with no crossing edges. Given a planar graph $G$, our prover provides the nodes with a distributed proof that there is $T$-embedding of $G$ in which no cotree-edges cross. 

We also show that our proof-labeling scheme has certificates of optimal size, in the sense that there are no \emph{locally checkable proofs}~\cite{GoosS16} for planarity that use certificates of size $o(\log n)$ bits (locally checkable proofs are verification mechanisms stronger than proof-labeling schemes). 
In fact, we show a more general result regarding the class of graphs excluding a complete graph $K_k$ as a minor, for any $k\geq 3$, and the class of graphs excluding a complete bipartite graph $K_{p,q}$ as a minor, for any $p,q\geq 2$. The proof for graphs excluding  $K_k$ is an extension of the technique used in~\cite{FeuilloleyH18} for lower bounding the size of \emph{global} certificates. The proof for graphs excluding $K_{p,q}$ is a non-standard adaptation of the original proof of a lower bound for checking spanning tree and leader election in~\cite{GoosS16}. Combining these results allows us to consider the class of graphs  excluding both $K_5$ and $K_{3,3}$ as minors, from which the lower bound for planarity follows directly. 

To sum up, the paper is dedicated to establishing the following results. 

\begin{theorem}\label{theo:main}
There is a 1-round proof-labeling scheme for planarity with certificates on $O(\log n)$ bits in $n$-node networks. 
\end{theorem} 

For a finite family $\cal{H}$ of graphs, let $\Forb(\mathcal{H})$ be the class of graphs with all the graphs in~$\mathcal{H}$ excluded as ``forbidden minors''.

\begin{theorem}\label{theo:lwb}
Let $\mathcal{F} = \{K_k:k\geq 3\}\cup\{K_{p,q}: p,q\geq 2\}$. For any non-empty finite family $\cal{H}$ of graphs in~$\mathcal{F} $, there are no locally checkable proofs for $\Forb(\cal{H})$ using certificates on $o(\log n)$ bits. 
\end{theorem} 

The following result follows directly from Theorem~\ref{theo:lwb} as the planar graphs form the class $\Forb(\{K_5,K_{3,3}\})$ by Wagner's theorem (see \cite{Diestel2000}). 

\begin{corollary}
There are no locally checkable proofs for planarity, using certificates on $o(\log n)$ bits. 
\end{corollary}

Theorem~\ref{theo:lwb} can also be used for other graph classes. For example, outerplanar graphs, which are the planar graphs that can be drawn in the plane with all the vertices on the outerface, form the class $\Forb(\{K_4,K_{2,3}\})$, and thus have no locally checkable proof with certificates on $o(\log n)$ bits.

\paragraph{Remark.} All our lower bounds hold even if one allows verification algorithms performing an arbitrarily large constant number of rounds.
\subsection{Related work}
\label{subsec:related-work}

\paragraph{ \bf Distributed algorithms for planar graphs.} 

As for the case of sequential computing, planar graphs have attracted significant interest in the context of distributed computing. Indeed, planar graphs do have structural properties that allow for fast distributed algorithms, in particular as far as approximation algorithms for classic problems such as minimum dominating set, or maximum matching are concerned. We refer to~\cite{CHS06, CHSWW14, CHSWW17,CHW08-spaa, CHW08-disc, HilkeLS13, LenzenOW08, LenzenPW13, Wawrzyniak14} for a non-exhaustive list of examples of such contributions. 

Recently,  it was proved that a combinatorial planar embedding that consists of each node knowing the clockwise order of its incident edges in a fixed planar drawing can be computed efficiently in the \textsf{CONGEST} model~\cite{GhaffariH16a}, and then used to derive $O(D)$-round distributed algorithms for MST and min-cut in planar networks of diameter~$D$~\cite{GhaffariH16b}, hence bypassing the lower bounds $\tilde{\Omega}(D+\sqrt{n})$ for general networks~\cite{PelegR00,SarmaHKKNPPW12}. Even more recently, a randomized distributed algorithm for computing a DFS tree in planar networks has been designed~\cite{GhaffariP17}, whose complexity beats the best known complexity bound $O(n)$ for general graphs~\cite{Awerbuch85}. 

More generally, it is  worth to mention that there is also a large recent literature on distributed algorithms designed for other families of sparse graphs, beyond planar graphs. This is for instance the case of graphs of bounded genus, and graphs of bounded expansion (see, e.g., \cite{AmiriMRS18,AmiriSS19}). We refer to \cite{Feuilloley20} for a bibliography on distributed algorithms in classes of sparse graphs. 

\paragraph{\bf Distributed decision and verification.}

Locally checkable proofs (LCPs) were introduced in~\cite{GoosS16}, as an extension of the seminal notion of \emph{proof-labeling scheme} (PLS) introduced in~\cite{KormanKP10}. There are only two differences between these two concepts, but they are subtle, and require some care. The first  difference is that LCPs allow for verification procedures performing many rounds, while PLS are restricted to a single round of verification. Nevertheless, PLS can be easily extended to many rounds whenever the certificates are large enough to contain IDs, which enables  to trade longer verification time for smaller  certificates size (see, e.g.,~\cite{FeuilloleyFHPP18}).  The second difference is more profound. While PLSs impose verification procedures exchanging the certificates only, LCPs allow the verification procedures to exchange additional information, e.g., the entire states of the nodes. 

As a consequence, a PLS with certificates on $O(\log n)$ bits can be systematically implemented in one round in the \textsf{CONGEST} model, while this is not necessarily the case of LCPs, if the additional information is of size $\omega(\log n)$ bits. Nevertheless, this phenomenon has little impact on upper bounds in general, as the states of the nodes are often encoded on a logarithmic number of bits too (e.g., a constant number of pointers to neighbors, a constant number of bits marking the nodes, etc.). On the other hand, the difference between LCPs and PLSs has strong impact on the design of lower bounds, especially for demonstrating impossibility results when using sub-logarithmic certificates, as the IDs of the nodes (which are part of their states) are seen by the neighbors in LCPs, but not in PLSs with certificates on $o(\log n)$ bits.  \emph{Non-deterministic local decision} (NLD)~\cite{FraigniaudKP13} is yet another similar notion of distributed certification. It differs from the PLS and LCP in the fact that the certificates must be independent from the identity-assignment to the nodes. 
We refer to~\cite{Feuilloley2019,FeuilloleyF16} for recent surveys on distributed decision and verification.

\paragraph{Remark.} Observe that the upper bound in Theorem~\ref{theo:main} is obtained by designing a PLS for planarity, while the lower bounds in Theorem~\ref{theo:lwb} hold even for LCPs.\medskip

All the aforementioned notions were extended by allowing the verifier to be randomized (see~\cite{FraigniaudPP19}). Such protocols were originally referred to as \emph{randomized PLS} (RPLS), but are nowadays referred to as distributed Merlin-Arthur (\textsf{dMA}) protocols. The concept of distributed verification has also been extended~\cite{BalliuDFO17,FeuilloleyFH16} in a way similar to the way NP was extended to the complexity classes forming the Polynomial Hierarchy, by alternating quantifiers. 

\paragraph{\bf Distributed interactive proofs.} 

Recently, \emph{distributed interactive proofs} were formalized~\cite{KolOS18}, and the classes $\mathsf{dAM}[k]$ and $\mathsf{dMA}[k]$, $k\geq 1$, were defined, where $k$ denotes the number of alternations between the centralized Merlin, and the decentralized Arthur. For instance, $\mathsf{dAM}[3]=\mathsf{dMAM}$ and $\mathsf{dMA}[2]=\mathsf{dMA}$, while LCP and PLS can be viewed as equal to $\mathsf{dAM}[1]=\mathsf{dM}$ (Merlin provides the nodes with their certificates, without challenges from Arthur). Distributed interactive protocols for problems like the existence of a non-trivial automorphism (\textsf{AUT}), and non-isomorphism ($\overline{\mathsf{ISO}}$) were designed and analyzed in~\cite{KolOS18}. The follow up paper~\cite{NaorPY20} improved the complexity of some of the protocols in~\cite{KolOS18}, either in terms of the number of interactions between the prover and the verifier, and/or in terms of the size of the certificates. A sophisticated generic way for constructing distributed IP protocols based on sequential IP protocols is presented in~\cite{NaorPY20}. One of the main outcome of this latter construction is a $\mathsf{dMAM}$ protocol for planarity, using certificates on $O(\log n)$ bits. For other recent results on distributed interactive proof, see~\cite{CrescenziFP19,FraigniaudMORT19}.

\section{Model and definitions}

We consider the standard model for distributed network computing~\cite{Peleg00}. The network is modeled as a simple connected graph $G=(V,E)$ --- note that self-loops and multiple edges can be eliminated without impacting planarity, and the planarity test could  trivially be performed independently in each connected components if the graphs were not connected. The number of nodes is denoted by $n=|V|$. Each node~$v$ has an identifier $\id(v)$, which is unique in the network, and picked from a range of IDs polynomial in~$n$. Therefore, every node identifier can be stored on $O(\log n)$ bits. 

We recall the notion of \emph{proof-labeling schemes}, and \emph{locally checkable proofs} for certifying graph classes. (These mechanisms can also be used to certify classes of node- or edge-labeled graphs, but we are solely interested in unlabeled graph classes in this paper). Let $\cal{C}$ be a graph class, e.g., planar graphs. A locally checkable proof for $\cal{C}$ is a prover-verifier pair where the \emph{prover} is a non-trustable  oracle assigning \emph{certificates} to the nodes, and the \emph{verifier} is a distributed algorithm enabling the nodes to check the correctness of the certificates by performing a single round of communication with their neighbors. Note that the certificates may not depend on the instance $G$ only, but also on the identifiers assigned to the nodes. In proof-labeling schemes, the information exchanged between the nodes during the verification phase is limited to the certificates. Instead, in locally checkable proofs, the nodes may exchange extra-information regarding their individual state (e.g., their IDs, if not included in the certificates, which might be the case for certificates of sub-logarithmic size). The prover-verifier pair must satisfy the following two properties.

\begin{description}
\item[Completeness:] Given $G\in\cal{C}$, the non-trustable prover can assign certificates to the nodes such that the verifier \emph{accepts} at all nodes;
\item[Soundness:] Given $G\notin\cal{C}$, for every certificate assignment to the nodes by the non-trustable prover, the verifier  \emph{rejects} in at least one node. 
\end{description}

The main \emph{complexity measure} for both locally checkable proofs, and proof-labeling schemes is the size of the certificates assigned to the nodes by the prover. 
 
 \medskip
 
 As an example, let us consider the class of paths. A possible proof-labeling scheme is as follows. Given a path $P=(v_1,\dots,v_n)$, the prover assigns the certificate $c(v_i)$ to node $v_i$, simply defined as the hop-distance between $v_1$ and $v_i$, for $i=1,\dots,n$. The verifier executed at node $u$ checks that $u$ has degree~1 or~2, and, if $u$ has degree~2, then it also checks that one of its two neighbors has certificate $c(u)-1$ while the other has certificate $c(u)+1$. If all tests are passed, then $u$ accepts, else it rejects. Completeness holds by construction. For soundness, we simply observe that, if all nodes accept, then the nodes are necessarily forming a path as we assume the network to be connected. 
 
Certifying the class of graphs \emph{containing} a path of length at least~$k$, for some fixed $k\geq 1$ can be done similarly, by also providing in the certificate of a node the IDs of its predecessor and successor in the path. The extremity of the path with positive hop-distance also checks that its hop-distance is at least~$k$. In addition, the certificates must contain a distributed proof that the path exists somewhere in the network. This is achieved by asking the prover to provide the nodes with a distributed encoding of a tree spanning all nodes, rooted at one origin of the path, plus a distributed proof for this spanning tree. Such proof is known for long, as it is implicitly or explicitly present in the early work on self-stabilizing algorithm~\cite{AfekKY97,AwerbuchPV91,ItkisL94} --- it simply consists of yet another hop-distance counter, plus the ID of the root. 

The examples above are the main ingredients enabling the design of a proof-labeling scheme for certifying \emph{non planarity}, whose existence is folklore in the context of distributed certification. It is indeed sufficient to provide the nodes with a distributed proof that the graph contains a subdivided $K_5$ or a subdivided~$K_{3,3}$. In both cases, this can be done by encoding the paths corresponding to these subdivided graphs  in the certificates of the nodes in these paths. In addition, a spanning tree and its proof are given in the certificates of all nodes for establishing the existence of the subdivided graph (every node receives a pointer to a parent, its hop-distance to the root, and the ID of the root, which must be a node of the subdivided $K_5$, or subdivided~$K_{3,3}$). We omit all the tedious details, and we expect that the reader has now understood the concept of  proof-labeling schemes (and locally checkable proofs). 

Observe that, in all the examples mentioned in this section, including certifying non-planarity, all certificates can be encoded on $O(\log n)$ bits. In the remaining part of the paper, we demonstrate that planarity can be distributedly certified, with $O(\log n)$-bit certificates too. 

\section{Upper bound}

In this section, we establish our main result, namely, that there is a 1-round proof-labeling scheme for planarity with certificates on $O(\log n)$ bits in $n$-node networks. The design and analysis of our proof-labeling scheme for planarity is decomposed into three stages. First, in Section~\ref{subsec:path-outerplanar}, we describe a proof-labeling scheme for a specific class of planar graphs, called \emph{path-outerplanar} graphs. Roughly, these graphs are Hamiltonian graphs which can be drawn in the plane without crossing edges in such a way that the Hamiltonian path forms a line, and the edges not in the path, are all on the same side of the line. Note that a non-Hamiltonian tree is outerplanar, but is not path-outerplanar. 
Next, in Section~\ref{se:transform}, 
we show how to transform a planar graph into a path-outerplanar graph.
More precisely, we show how to transform a $T$-embedding~\cite{FraysseixR85} of a graph into a drawing of a new graph such that, roughly, the new graph is path-outerplanar if and only if the original graph is planar. 
Finally, Section~\ref{se:PLS-planar} shows how to use this transformation for extending the proof-labeling scheme for path-outerplanarity to a proof-labeling scheme for planarity. Throughout the scheme, we use the key property that every planar graph is 5-degenerate (i.e., every subgraph has a vertex of degree at most~5), in order to distribute the certificates evenly among the nodes.

\subsection{Certifying path-outerplanar graphs}
\label{subsec:path-outerplanar}

In this section, we present a proof-labeling scheme for a subclass of outerplanar graphs, that will be used as a building block for our scheme for planar graphs.  Recall that a graph is outerplanar if it has a planar drawing with all vertices incident to the same face, called the outerface. We start with a combinatorial definition of path-outerplanarity, and then show the equivalence with a geometric definition.

\begin{definition}\label{def:path-outerplanar}
A graph $G = (V,E)$ is \emph{path-outerplanar} if there is a total ordering $P=(V,<)$ of its vertices such that $P$ forms a path (i.e., consecutive vertices in $P$ are adjacent in~$G$), and, for any pair of edges $\{a,b\}, \{c,d\} \in E$ with $a<b$ and $c<d$, one of the following  inequalities holds: $a<b \leq c<d$, $c<d \leq a<b$, $a \leq c < d \leq b$, or $c \leq a < b \leq d$. The ordering $P$ is called a \emph{path-outerplanarity witness} of $G$.
\end{definition} 

\begin{lemma}\label{lem:path-outerplanar-drawing}
A graph is \emph{path-outerplanar} if and only if it has a Hamiltonian path that can be drawn as a horizontal line such that all edges that do not belong to the path can be drawn above that line as semi-circles without crossings. 
\end{lemma}

\begin{proof}
Assume that $G$ is path-outerplanar with witness $(1,\dots, n)$. The vertices $\{1,\dots,n\}$ can be drawn on a horizontal line, with vertex~$i$ at coordinate~$i$. Each edge $\{a,b\}$ of $G$ can be drawn as a semi-circle, with endpoints $a$ and $b$, above the line. The drawing is planar since, by definition, for every two edges $\{a,b\}, \{c,d\}$ of $G$ with $a<b$ and $c<d$, we have $|[a,b]\cap[c,d]|\leq 1$, or $[a,b]\subset [c,d]$, or $[c,d]\subset [a,b]$. 
Conversely, let us assume that there exists a planar drawing as described. Let us then consider the total order induced by the placement of the vertices along this horizontal path (say, from left to right). Since no two semi-circles cross, for every pair $\{a,b\}, \{c,d\}$ of edges, one of the four inequalities in Definition~\ref{def:path-outerplanar} must be satisfied, and thus the graph is path-outerplanar.
\end{proof}

\begin{lemma}\label{lem:PLS-path-outerplanar}
There is a 1-round proof-labeling scheme for path-outerplanarity, with certificates on $O(\log n)$ bits in $n$-node networks.
\end{lemma} 

\begin{proof}
Given a path-outerplanar graph $G$, the prover computes a witness $P=(v_1,\dots,v_n)$ for~$G$, and sends an $O(\log n)$-bit certificate to every vertex~$x$, with the following information:
\begin{enumerate}
\item the number $n$ of vertices of the graph;
\item the rank of $x$ in $P$, i.e., the value $i$ such that $x=v_i$;
\item the shortest interval $I(x) = [a,b]$ such that $\{v_a,v_b\}$ is an edge of G, and $a<i<b$ (if no such interval exists, then the prover sets $I(x)$ to $[0,n+1]$).
\end{enumerate}
In addition, the prover  provides the nodes with information in their certificates enabling them to check that $n$ is indeed the correct number of vertices, that the vertices are correctly ranked from $1$ to $n$, and that this ranking induces a path. We do not detail these parts of the certificates, nor we describe the verification algorithm for these certificates, as this can be achieved using standard techniques~\cite{KormanKP10} using $O(\log n)$-bit certificates. 

For simplicity, in the algorithm as well as in the remaining part of the proof, we denote the vertices by their ranks, as numbers from~$1$ to~$n$. We also add two virtual vertices $0$ and $n+1$, with virtual edges $\{0,1\}, \{0,n+1\}$, and $\{n,n+1\}$. We set  $I(0)=I(n+1)=[-\infty, \infty]$. This ensures that any vertex $x$ with $1 \leq x \leq n$ has at least one neighbor smaller than itself, and another larger than itself. Of course, the verification algorithm is only performed at the real vertices, i.e., those from $1$ to $n$, with node~$1$ (resp., node~$n$)  simulating the behavior of its virtual neighbor~$0$ (resp.,~$n+1$).  

The verification algorithm performed at each node $x$ is displayed in Algorithm~\ref{al:path-outerplanar}. 

\begin{figure}[h]
\begin{center}
\small
\begin{minipage}{.86\linewidth}
\begin{algorithm}[H]
\caption{\small Verification procedure for path-outerplanarity at node $x$, $1 \leq x \leq n$.}
\label{al:path-outerplanar}
let $x^-_\ell <  \dots< x^-_0 < x^+_0 < \dots < x^+_k$, with $\ell\geq 0$ and $k\geq 0$, be the neighbors of $x$; \\
collect the certificates of each neighbor: graph size, rank, and interval;  \\
check that the ranks correspond to a spanning path of size $n$; \label{l:tPath}  \\
let $I(x) = [a,b]$;  \\
check that $a < x < b$, and that all neighbors of $x$ are in the interval $[a,b]$;\label{l:tInt}  \\
\For{$i=0$ \textrm{\bf to} $k-1$}
{
	check that $I(x^+_i)=[x,x^+_{i+1}]$;\label{l:tRight}
}	
\For{$i=0$ \textrm{\bf to} $\ell-1$}
{
	check that $I(x^-_i)=[x^-_{i+1},x$]\label{l:tLeft};
}
\If{$x^+_k < b$}
{
	check that $I(x^+_k)=[a,b]$\label{l:tRight-ab};
}
\If{$x^-_\ell > a$}
{
	check that $I(x^-_\ell)=[a,b]$\label{l:tLeft-ab};
}
\For{\rm every neighbor $y$ of $x$}
{
	\If{\rm one of the endpoints of $I(y)$ is $x$}
	{
		check that the other endpoint of $I(y)$ is adjacent to $x$\label{l:tExtreme1};\\
		check that $I(y) \subsetneq I(x)$\label{l:tExtreme2};
	}
}
\textbf{if} all checks are passed \textbf{then} accept \textbf{else} reject.
\end{algorithm}
\end{minipage}
\end{center}
\end{figure}

We first show completeness. 

\begin{claim}
If $G$ is path-outerplanar, and if the prover provides certificates corresponding to a witness $P$, then Algorithm~\ref{al:path-outerplanar} accepts at all nodes.
\end{claim} 

\begin{proofofclaim}
For any node $x$, let $x^-_\ell <  \dots< x^-_0 < x^+_0 < \dots < x^+_k$, with $\ell\geq 0$ and $k\geq 0$, be the neighbors of $x$ in the ordering $P$. In particular, $x^-_0 < x < x^+_0$. 
Figure~\ref{fig:path-outerplanar} illustrates the structure of the neighborhood of a node~$x$, and the interval $I(x)=[a,b]$ (cf. Lemma~\ref{lem:path-outerplanar-drawing}). Note that, since $G$ is planar, $a\leq x^-_\ell$, and $b\geq x^+_k$.
The test of Line~\ref{l:tPath} succeeds since $P$ is the witness for~$G$. Observe that the vertices $y$ with $I(y)=[a,b]$ are exactly the vertices incident to the unique face below the edge~$\{a,b\}$, except $a$ and $b$ --- this also holds for the nodes $0$ and $n+1$, as we could add  a virtual edge between $-\infty$ and~$\infty$. In particular, no edge incident to $x$ can cross the edge $\{a,b\}$, which implies that the test of  Line~\ref{l:tInt} is passed. The same arguments applied to $x^+_i$ explains why the test of Line~\ref{l:tRight} is passed (respectively, $x^-_i$ and Line~\ref{l:tLeft}). For the rightmost neighbor $x^+_k$ of $x$, we distinguish between the cases  $x^+_k < b$, and $x^+_k = b$. In the first case, $x_k$ is also incident to the face below the edge $\{a,b\}$, and thus the test of Line~\ref{l:tRight-ab} succeeds. The same arguments apply for the test at Line~\ref{l:tLeft-ab}. The case $x_k=b$ is not explicitly addressed by Algorithm~\ref{al:path-outerplanar} applied on $x$. However, it is part of the tests performed at Lines~\ref{l:tExtreme1} and~\ref{l:tExtreme2}, when applied to node $x^+_k$. These two tests succeed by definition of $I(x)$ and $I(y)$. It follows that node~$x$ accepts, as desired. 
 \end{proofofclaim}
 
 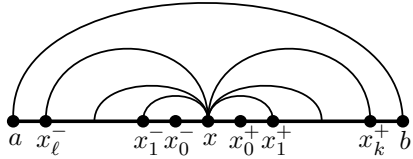
\begin{figure}[htb]
\begin{center}
\scalebox{0.9}{
\tikzset{every picture/.style={line width=0.75pt}} 

\begin{tikzpicture}[x=0.75pt,y=0.75pt,yscale=-0.6,xscale=0.6]

\draw [line width=1.5]    (155,210.25) -- (515,210.25) ;

\draw  [fill={rgb, 255:red, 0; green, 0; blue, 0 }  ,fill opacity=1 ] (480,210.25) .. controls (480,207.49) and (482.24,205.25) .. (485,205.25) .. controls (487.76,205.25) and (490,207.49) .. (490,210.25) .. controls (490,213.02) and (487.76,215.25) .. (485,215.25) .. controls (482.24,215.25) and (480,213.02) .. (480,210.25) -- cycle ;
\draw  [fill={rgb, 255:red, 0; green, 0; blue, 0 }  ,fill opacity=1 ] (390,210.25) .. controls (390,207.49) and (392.24,205.25) .. (395,205.25) .. controls (397.76,205.25) and (400,207.49) .. (400,210.25) .. controls (400,213.02) and (397.76,215.25) .. (395,215.25) .. controls (392.24,215.25) and (390,213.02) .. (390,210.25) -- cycle ;
\draw  [fill={rgb, 255:red, 0; green, 0; blue, 0 }  ,fill opacity=1 ] (360,210.25) .. controls (360,207.49) and (362.24,205.25) .. (365,205.25) .. controls (367.76,205.25) and (370,207.49) .. (370,210.25) .. controls (370,213.02) and (367.76,215.25) .. (365,215.25) .. controls (362.24,215.25) and (360,213.02) .. (360,210.25) -- cycle ;
\draw  [fill={rgb, 255:red, 0; green, 0; blue, 0 }  ,fill opacity=1 ] (330,210.25) .. controls (330,207.49) and (332.24,205.25) .. (335,205.25) .. controls (337.76,205.25) and (340,207.49) .. (340,210.25) .. controls (340,213.02) and (337.76,215.25) .. (335,215.25) .. controls (332.24,215.25) and (330,213.02) .. (330,210.25) -- cycle ;
\draw  [fill={rgb, 255:red, 0; green, 0; blue, 0 }  ,fill opacity=1 ] (300,210.25) .. controls (300,207.49) and (302.24,205.25) .. (305,205.25) .. controls (307.76,205.25) and (310,207.49) .. (310,210.25) .. controls (310,213.02) and (307.76,215.25) .. (305,215.25) .. controls (302.24,215.25) and (300,213.02) .. (300,210.25) -- cycle ;
\draw  [fill={rgb, 255:red, 0; green, 0; blue, 0 }  ,fill opacity=1 ] (270,210.25) .. controls (270,207.49) and (272.24,205.25) .. (275,205.25) .. controls (277.76,205.25) and (280,207.49) .. (280,210.25) .. controls (280,213.02) and (277.76,215.25) .. (275,215.25) .. controls (272.24,215.25) and (270,213.02) .. (270,210.25) -- cycle ;
\draw  [fill={rgb, 255:red, 0; green, 0; blue, 0 }  ,fill opacity=1 ] (180,210.25) .. controls (180,207.49) and (182.24,205.25) .. (185,205.25) .. controls (187.76,205.25) and (190,207.49) .. (190,210.25) .. controls (190,213.02) and (187.76,215.25) .. (185,215.25) .. controls (182.24,215.25) and (180,213.02) .. (180,210.25) -- cycle ;
\draw  [fill={rgb, 255:red, 0; green, 0; blue, 0 }  ,fill opacity=1 ] (510,210.25) .. controls (510,207.49) and (512.24,205.25) .. (515,205.25) .. controls (517.76,205.25) and (520,207.49) .. (520,210.25) .. controls (520,213.02) and (517.76,215.25) .. (515,215.25) .. controls (512.24,215.25) and (510,213.02) .. (510,210.25) -- cycle ;
\draw  [fill={rgb, 255:red, 0; green, 0; blue, 0 }  ,fill opacity=1 ] (150,210.25) .. controls (150,207.49) and (152.24,205.25) .. (155,205.25) .. controls (157.76,205.25) and (160,207.49) .. (160,210.25) .. controls (160,213.02) and (157.76,215.25) .. (155,215.25) .. controls (152.24,215.25) and (150,213.02) .. (150,210.25) -- cycle ;
\draw [color={rgb, 255:red, 0; green, 0; blue, 0 }  ,draw opacity=1 ]   (155,205.25) .. controls (154.17,64.75) and (513.5,63.42) .. (515,210.25) ;

\draw    (185,210.25) .. controls (185,120.75) and (335.67,119.42) .. (335,210.25) ;

\draw    (335,210.25) .. controls (335,120.75) and (484.33,119.42) .. (485,210.25) ;

\draw    (275,210.25) .. controls (274.33,178.75) and (335.67,178.75) .. (335,210.25) ;

\draw    (335,210.25) .. controls (334.33,178.75) and (395.67,178.75) .. (395,210.25) ;

\draw    (230,210.25) .. controls (230.33,165.42) and (336.33,166.75) .. (335,210.25) ;

\draw    (335,210.25) .. controls (335.33,165.42) and (441.33,166.75) .. (440,210.25) ;

\draw (337,227.25) node    {$x$};
\draw (369,227.75) node    {$x^{+}_{0}$};
\draw (399,227.75) node    {$x^{+}_{1}$};
\draw (489,227.75) node    {$x^{+}_{k}$};
\draw (517,227.25) node    {$b$};
\draw (157,227.25) node    {$a$};
\draw (309,227.75) node    {$x^{-}_{0}$};
\draw (281,227.75) node    {$x^{-}_{1}$};
\draw (191,227.75) node    {$x^{-}_{\ell}$};

\end{tikzpicture}
}
\vspace{-2ex}
\caption{The neighborhood of a node $x$, and the edge $\{a,b\}$ covering $x$.}\label{fig:path-outerplanar}
\end{center}
\end{figure}

It remains to establish soundness. Namely, we prove that, if Algorithm~\ref{al:path-outerplanar} accepts at all nodes, then $G$ is path-outerplanar. As mentioned before, we do not detail the test for the spanning path (Line~\ref{l:tPath}), but this test ensures that the ranks are consistent, from $1$ to $n$, and that consecutive vertices  in this order are adjacent in $G$.  We thus safely denote the vertices by their ranks.  In fact, we show that if Algorithm~\ref{al:path-outerplanar} accepts at all nodes, then $G$ is path-outerplanar with witness $P = (1,2,\dots,n)$. 
Let us assume, for the purpose of contradiction, that this is not the case. It follows that there are four nodes $x<y<z<t$ such that both edges $\{x,z\}$ and $\{y,t\}$ appear in $G$. Under these conditions, let us choose such four nodes with the additional conditions that (i)~$z-y$ is minimum, and (ii)~$t-x$ is minimum subject to~(i). The following claim is a straightforward consequence of conditions~(i) and~(ii) --- see Figure~\ref{fig:interlaced} for an illustration of the framework of the claim. 

\begin{figure}[htb]
\begin{center}
\scalebox{0.8}{
\tikzset{every picture/.style={line width=0.75pt}} 

\begin{tikzpicture}[x=0.75pt,y=0.75pt,yscale=-0.6,xscale=0.6]

\draw [line width=1.5]    (151,200.25) -- (511,200.25) ;

\draw  [fill={rgb, 255:red, 0; green, 0; blue, 0 }  ,fill opacity=1 ] (225,200) .. controls (225,197.24) and (227.24,195) .. (230,195) .. controls (232.76,195) and (235,197.24) .. (235,200) .. controls (235,202.76) and (232.76,205) .. (230,205) .. controls (227.24,205) and (225,202.76) .. (225,200) -- cycle ;
\draw  [fill={rgb, 255:red, 0; green, 0; blue, 0 }  ,fill opacity=1 ] (506,200.25) .. controls (506,197.49) and (508.24,195.25) .. (511,195.25) .. controls (513.76,195.25) and (516,197.49) .. (516,200.25) .. controls (516,203.02) and (513.76,205.25) .. (511,205.25) .. controls (508.24,205.25) and (506,203.02) .. (506,200.25) -- cycle ;
\draw  [fill={rgb, 255:red, 0; green, 0; blue, 0 }  ,fill opacity=1 ] (146,200.25) .. controls (146,197.49) and (148.24,195.25) .. (151,195.25) .. controls (153.76,195.25) and (156,197.49) .. (156,200.25) .. controls (156,203.02) and (153.76,205.25) .. (151,205.25) .. controls (148.24,205.25) and (146,203.02) .. (146,200.25) -- cycle ;
\draw [color={rgb, 255:red, 0; green, 0; blue, 0 }  ,draw opacity=1 ]   (151,195.25) .. controls (150.17,54.75) and (428.5,53.17) .. (430,200) ;

\draw [line width=1.5]  [dash pattern={on 1.69pt off 2.76pt}]  (180,200) .. controls (180,110.5) and (301.67,109.42) .. (301,200.25) ;

\draw [line width=1.5]  [dash pattern={on 1.69pt off 2.76pt}]  (310,200) .. controls (310.33,155.17) and (391.33,156.5) .. (390,200) ;

\draw [color={rgb, 255:red, 0; green, 0; blue, 0 }  ,draw opacity=1 ]   (230,200) .. controls (229.17,59.5) and (509.5,53.42) .. (511,200.25) ;

\draw  [fill={rgb, 255:red, 0; green, 0; blue, 0 }  ,fill opacity=1 ] (425,200) .. controls (425,197.24) and (427.24,195) .. (430,195) .. controls (432.76,195) and (435,197.24) .. (435,200) .. controls (435,202.76) and (432.76,205) .. (430,205) .. controls (427.24,205) and (425,202.76) .. (425,200) -- cycle ;
\draw [line width=1.5]  [dash pattern={on 1.69pt off 2.76pt}]  (361,200.25) .. controls (361,110.75) and (481.67,109.42) .. (481,200.25) ;

\draw [line width=1.5]  [dash pattern={on 1.69pt off 2.76pt}]  (270,200) .. controls (270.33,155.17) and (351.33,156.5) .. (350,200) ;

\draw (513,220) node    {$t$};
\draw (153,217.25) node    {$x$};
\draw (233,220) node    {$y$};
\draw (433,220) node    {$z$};

\end{tikzpicture}
}
\end{center}
\vspace{-5ex}
\caption{The crossings involving edges depicted by dotted lines cannot exist.}\label{fig:interlaced}
\end{figure}
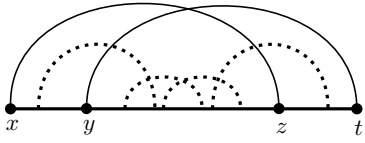

\begin{claim}\label{cl:non-cr}
For every edge $\{a,b\}$, if  $y<a<z$, then $y\leq b \leq z$. Also, if $y=a<b<t$, then $b\leq z$, and if $x< a <z=b$, then $a\geq y$. Finally, any two edges with both endpoints in the interval $[y,z]$ are non-crossing. \hfill $\diamond$
\end{claim}

To understand the structure of the graph in between $y$ and $z$, let us define:
\begin{itemize}
\item $y^+_j$, the rightmost neighbor of $y$ strictly before $t$ in the ordering $P$,
\item $z^-_i$, the leftmost neighbor of $z$  strictly after $x$ in the ordering $P$.
\end{itemize}

Note that these two nodes exist, as the node immediately on the right of $y$ in the path appears strictly before~$t$, and the node immediately on the left of $z$ in the path appears strictly after~$x$. 

\begin{claim}\label{claim:inequality-intervals}
The following holds: $y < y_j^+ < z^-_i <z $, $I(y^+_j)= [y,t]$, and $I(z^-_i)=[x,z]$.
\end{claim}

\begin{proofofclaim}
 By Claim~\ref{cl:non-cr}, the inequality $y < y^+_j \leq z$ holds. Also, as the test of Line~\ref{l:tRight} applied to vertex $y$ succeeds, it follows that $I(y^+_j)= [y,t]$. 
Observe that $\{y,z\}$ cannot be an edge of the graph. Indeed, if $\{y,z\}$ is an edge of the graph, then it must be the case that $z = y^+_j$, and thus $I(z)=[y,t]$. Also, symmetrically, the equality $I(y)=[x,z]$ must hold too. As the test of Line~\ref{l:tExtreme2} applied to vertex $y$ succeeds, it follows that $I(z) \subsetneq I(y)$, which contradicts the fact that the two intervals overlap. 
Therefore $\{y,z\}$ is not an edge of the graph. It follows that $y^+_j < z$. By the same arguments as for $y^+_j$, but applied to the left-hand side of $z$, it follows that $y<z^-_i$, and $I(z^-_i)=[x,z]$. Also observe that we have $y^+_j \leq z^-_i $, by Claim~\ref{cl:non-cr} applied to edges $\{y,y^+_j\}$ and $\{z^-_i,z\}$. 
Finally, as $I(y^+_j)= [y,t]$ and $I(z^-_i)=[x,z]$, and as these two intervals are different, it must be that $y^+_j<z^-_i $.
\end{proofofclaim}

In the following, an edge $\{u,v\}$  is said to be \emph{maximal} if $v$ is the rightmost neighbor of $u$, and $u$ is the leftmost neighbor of~$v$. See Figure~\ref{fig:extreme-edge}.

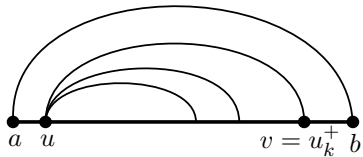
\begin{figure}[htb]
\begin{center}
\scalebox{0.9}{
\tikzset{every picture/.style={line width=0.75pt}} 

\begin{tikzpicture}[x=0.75pt,y=0.75pt,yscale=-0.6,xscale=0.6]

\draw [line width=1.5]    (171,210) -- (485,210) ;

\draw  [fill={rgb, 255:red, 0; green, 0; blue, 0 }  ,fill opacity=1 ] (435,210) .. controls (435,207.24) and (437.24,205) .. (440,205) .. controls (442.76,205) and (445,207.24) .. (445,210) .. controls (445,212.76) and (442.76,215) .. (440,215) .. controls (437.24,215) and (435,212.76) .. (435,210) -- cycle ;
\draw  [fill={rgb, 255:red, 0; green, 0; blue, 0 }  ,fill opacity=1 ] (196,210) .. controls (196,207.24) and (198.24,205) .. (201,205) .. controls (203.76,205) and (206,207.24) .. (206,210) .. controls (206,212.76) and (203.76,215) .. (201,215) .. controls (198.24,215) and (196,212.76) .. (196,210) -- cycle ;
\draw  [fill={rgb, 255:red, 0; green, 0; blue, 0 }  ,fill opacity=1 ] (480,210) .. controls (480,207.24) and (482.24,205) .. (485,205) .. controls (487.76,205) and (490,207.24) .. (490,210) .. controls (490,212.76) and (487.76,215) .. (485,215) .. controls (482.24,215) and (480,212.76) .. (480,210) -- cycle ;
\draw  [fill={rgb, 255:red, 0; green, 0; blue, 0 }  ,fill opacity=1 ] (166,210) .. controls (166,207.24) and (168.24,205) .. (171,205) .. controls (173.76,205) and (176,207.24) .. (176,210) .. controls (176,212.76) and (173.76,215) .. (171,215) .. controls (168.24,215) and (166,212.76) .. (166,210) -- cycle ;
\draw [color={rgb, 255:red, 0; green, 0; blue, 0 }  ,draw opacity=1 ]   (171,205) .. controls (170.17,64.5) and (483.5,68.17) .. (485,215) ;

\draw    (201,210) .. controls (201.5,112) and (439.83,113.17) .. (440,210) ;

\draw    (201,210) .. controls (202.5,146) and (379.5,141) .. (380,210) ;

\draw    (201,210) .. controls (202.17,161.17) and (340.17,162.5) .. (340,210) ;

\draw (436.5,227.5) node    {$v=u^{+}_{k}$};
\draw (487,230) node    {$b$};
\draw (173,227) node    {$a$};
\draw (203,227) node    {$u$};

\end{tikzpicture}
}
\end{center}
\vspace{-5ex}
\caption{A maximal edge $\{u,v\}$.}\label{fig:extreme-edge}
\end{figure} 

The next claim describes the structure of the graph between $y^+_j$ and $z^-_i $ (see Figure~\ref{fig:path_max_edges}). 

\begin{figure}[htb]
\begin{center}
\scalebox{0.8}{
\tikzset{every picture/.style={line width=0.75pt}} 

\begin{tikzpicture}[x=0.75pt,y=0.75pt,yscale=-0.6,xscale=0.6]

\draw [line width=1.5]    (90,264.53) -- (520,264.53) ;

\draw  [fill={rgb, 255:red, 0; green, 0; blue, 0 }  ,fill opacity=1 ] (145,264.53) .. controls (145,261.77) and (147.24,259.53) .. (150,259.53) .. controls (152.76,259.53) and (155,261.77) .. (155,264.53) .. controls (155,267.29) and (152.76,269.53) .. (150,269.53) .. controls (147.24,269.53) and (145,267.29) .. (145,264.53) -- cycle ;
\draw [color={rgb, 255:red, 0; green, 0; blue, 0 }  ,draw opacity=1 ]   (510,120) .. controls (515.5,120) and (151.5,88) .. (150,264.53) ;

\draw [line width=0.75]    (150,264.53) .. controls (148.5,196.53) and (218.5,194.53) .. (220,264.53) ;

\draw [color={rgb, 255:red, 0; green, 0; blue, 0 }  ,draw opacity=1 ]   (461,264.53) .. controls (459.5,94) and (99.5,113) .. (100,110) ;

\draw  [fill={rgb, 255:red, 0; green, 0; blue, 0 }  ,fill opacity=1 ] (456,264.53) .. controls (456,261.77) and (458.24,259.53) .. (461,259.53) .. controls (463.76,259.53) and (466,261.77) .. (466,264.53) .. controls (466,267.29) and (463.76,269.53) .. (461,269.53) .. controls (458.24,269.53) and (456,267.29) .. (456,264.53) -- cycle ;
\draw [line width=2.25]    (220,264.53) .. controls (220.33,219.69) and (281.33,221.03) .. (280,264.53) ;

\draw [line width=0.75]    (391,264.53) .. controls (389.5,196.53) and (459.5,194.53) .. (461,264.53) ;

\draw [line width=2.25]    (280,264.53) .. controls (280.33,219.69) and (331.33,221.5) .. (330,265) ;

\draw [line width=2.25]    (330,265) .. controls (330.33,220.17) and (392.33,221.03) .. (391,264.53) ;

\draw  [fill={rgb, 255:red, 0; green, 0; blue, 0 }  ,fill opacity=1 ] (215,264.53) .. controls (215,261.77) and (217.24,259.53) .. (220,259.53) .. controls (222.76,259.53) and (225,261.77) .. (225,264.53) .. controls (225,267.29) and (222.76,269.53) .. (220,269.53) .. controls (217.24,269.53) and (215,267.29) .. (215,264.53) -- cycle ;
\draw  [fill={rgb, 255:red, 0; green, 0; blue, 0 }  ,fill opacity=1 ] (275,264.53) .. controls (275,261.77) and (277.24,259.53) .. (280,259.53) .. controls (282.76,259.53) and (285,261.77) .. (285,264.53) .. controls (285,267.29) and (282.76,269.53) .. (280,269.53) .. controls (277.24,269.53) and (275,267.29) .. (275,264.53) -- cycle ;
\draw  [fill={rgb, 255:red, 0; green, 0; blue, 0 }  ,fill opacity=1 ] (325,265) .. controls (325,262.24) and (327.24,260) .. (330,260) .. controls (332.76,260) and (335,262.24) .. (335,265) .. controls (335,267.76) and (332.76,270) .. (330,270) .. controls (327.24,270) and (325,267.76) .. (325,265) -- cycle ;
\draw  [fill={rgb, 255:red, 0; green, 0; blue, 0 }  ,fill opacity=1 ] (386,264.53) .. controls (386,261.77) and (388.24,259.53) .. (391,259.53) .. controls (393.76,259.53) and (396,261.77) .. (396,264.53) .. controls (396,267.29) and (393.76,269.53) .. (391,269.53) .. controls (388.24,269.53) and (386,267.29) .. (386,264.53) -- cycle ;
\draw  [dash pattern={on 0.84pt off 2.51pt}]  (220,180) -- (220,264.53) ;

\draw  [dash pattern={on 0.84pt off 2.51pt}]  (280,180) -- (280,264.53) ;

\draw  [dash pattern={on 0.84pt off 2.51pt}]  (330,180) -- (330,265) ;

\draw  [dash pattern={on 0.84pt off 2.51pt}]  (390,180) -- (391,264.53) ;

\draw (153,284.53) node    {$y$};
\draw (463,284.53) node    {$z$};
\draw (221,287.5) node    {$y^{+}_{j}$};
\draw (391,287.5) node    {$z^{-}_{i}$};
\draw (308.5,200) node  [font=\large,color={rgb, 255:red, 0; green, 0; blue, 0 }  ,opacity=1 ]  {$Q$};

\end{tikzpicture}
}
\end{center}
\vspace{-5ex}
\caption{The path $Q$ of maximal edges between $y$ and $z$.}\label{fig:path_max_edges}
\end{figure}
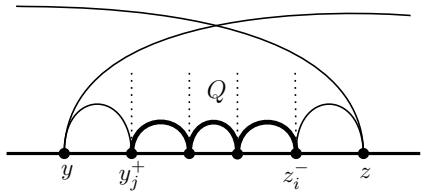

\begin{claim}\label{cl:path-max}
There is a path $Q=(q_1,\dots, q_s)$ in $G$, with $s > 1$, $q_1=y^+_j$, and $q_s = z^-_i$, such that,  for every $i=1,\dots,s-1$, $q_i<q_{i+1}$ and $\{q_i,q_{i+1}\}$ is a maximal edge.
\end{claim}

\begin{proofofclaim}
The path is constructed as follows: $q_1= y^+_j$ and, as long as the rightmost vertex $q_i$ of the partial path created so far satisfies  $q_i<z^-_i$, we pick $q_{i+1}$ as the rightmost neighbor of $q_i$. We claim that $q_{i+1} \leq z^-_i$. This is because:
\begin{enumerate}
\item by Claim~\ref{cl:non-cr}, it cannot be the case that $q_{i+1}>z$;
\item $q_{i+1}=z$ cannot hold because $q_i$ would then contradict the choice of $z^-_i$ as  leftmost neighbor of $z$ larger than $x$;
\item $z^-_i < q_{i+1}<z$  is impossible as the crossing edges $\{q_i,q_{i+1}\}$ and $\{z^-_i,z\}$ would contradict Claim~\ref{cl:non-cr}.
\end{enumerate}
Therefore the construction of $Q$ correctly terminates at vertex $z^-_i$.

It remains to show that edges $\{q_i,q_{i+1}\}$ are maximal for all $i=1,\dots,s-1$. By construction, $q_{i+1}$ is the rightmost neighbor of $q_i$, and thus it is sufficient to prove that $q_i$ is the leftmost neighbor of $q_{i+1}$. Assume that $q_{i+1}$ has a neighbor $q'<q_i$. As for the construction of $Q$, observe that $q' \geq y$ by Claim~\ref{cl:non-cr} applied to the edge $\{q',q_{i+1}\}$. Also, $q'$ cannot be one of the vertices $y\in\{q_1,\dots,q_i\}$ because,  for every $j=1,\dots,i$, $q_j$ is the rightmost neighbor of $q_{j-1}$ smaller than $z$. Eventually, $q'$ cannot be strictly between $q_{j-1}$ and $q_j$ for some $j$ with $1 \leq j \leq i$, by Claim~\ref{cl:non-cr} applied to edges $\{q',q_i\}$ and $\{q_{j-1},q_j\}$. Therefore, $q_i$ is the leftmost neighbor of $q_{i+1}$, and all edges $\{q_i,q_{i+1}\}$ are maximal.
\end{proofofclaim}

We were interested in maximal edges because of the following claim.

\begin{claim}\label{cl:extreme-edge}
If the edge $\{u,v\}$ is maximal, and Algorithm~\ref{al:path-outerplanar} accepts at all nodes, then $I(u) = I(v)$.
\end{claim}

\begin{proofofclaim}
We refer to Figure~\ref{fig:extreme-edge}  for an illustration of the arguments developed in the proof. Let $I(u)=[a,b]$ and $I(v)=[c,d]$. Let us assume, w.l.o.g., that $u <v$.  Since $\{u,v\}$ is maximal, $v=u^+_k$ where $u^+_k$ is the largest neighbor of $u$. Since Algorithm~\ref{al:path-outerplanar} accepts at all nodes, all tests performed by the algorithm at all nodes have positive outcome. Therefore, by the test of Line~\ref{l:tInt} applied to $x=u$, we derive that $u^+_k \leq b$, and thus $v\leq b$. Now let us consider three cases separately. 
\begin{enumerate}
\item If $v<b$, then the test of Line~\ref{l:tRight-ab} applied to $u$ guarantees that $I(u^+_k) = I(u)$, and the claim follows. 
\item If $v=b$ but $u = v^-_\ell>c$, where $v^-_\ell$ is the leftmost neighbor of $v$, then, by symmetry,  the test of Line~\ref{l:tLeft} applied to $x=v$ implies that $I(u) = I(v)$, as claimed.
\item If $v = b$ and $u = c$, then the test of Line~\ref{l:tExtreme2} applied at $x=u$, and at $x'=b=v$ implies that  $I(v) \subsetneq I(u)$. Symmetrically, by the same test at $x=v$ and $x'=u=c$, we also get that $I(u) \subsetneq I(v)$. This is a contradiction, and thus this case cannot appear.
\end{enumerate}   
This completes the proof of Claim~\ref{cl:extreme-edge}. 
\end{proofofclaim}

We can now complete the soundness proof for our proof-labeling scheme. 
By Claims~\ref{cl:path-max} and~\ref{cl:extreme-edge}, all the vertices of the path $Q$ (formed of maximal edges) were assigned the same interval~$I$. This contradicts the fact that, according to Claim~\ref{claim:inequality-intervals}, $I(q_1)=[y,t]$ and $I(q_s)=[x,z]$. This contradiction completes the proof of Lemma~\ref{lem:PLS-path-outerplanar}. 
\end{proof}

\subsection{From path-outerplanar  graphs to planar graphs}
\label{se:transform}

The previous section has demonstrated how to construct a compact proof-labeling scheme for the class of path-outerplanar graphs. 
For extending the scheme to planar graphs, this section presents a transformation that maps any planar graph to a path-outerplanar graph. 
To get an intuition of this transformation, let us first explain how a tree~$T$ can be transformed into a path. The transformation is inspired by the classical approximation algorithm by Christofides for the traveling salesman problem. 
A depth-first search (DFS) traversal of the tree returns an ordering of the nodes by indices from~1 to $2n-1$, where a same node can be assigned several indices. This ordering can be transformed into a path on $2n-1$ nodes, by creating a copy of each node for each time it is visited by the DFS traversal, and linking such nodes according to the ordering of the traversal. 

More formally, let $T$ be a tree spanning a planar graph $G$. We fix a drawing of $G$ in the plane without crossing edges (see Figure~\ref{fig:tree1}(a), ignoring node~$r'$ at this stage of the proof). We can assume, w.l.o.g., that all edges are drawn as straight lines (it is not crucial in the construction, but it helps for understanding it). The tree $T$ is rooted at an arbitrary vertex~$r$. For any vertex $v\neq r$, its neighbor on the path from $v$ to $r$ is called the \emph{parent} of~$v$, and the other neighbors of~$v$ are called its \emph{children}. For any node $v$, $\deg_T(v)$ is the degree of $v$ in $T$. A DFS traversal of  $T$  starting from $r$ provides a \emph{DFS-mapping} $f : \{1, \dots, 2n-1\} \rightarrow V(T)$  satisfying the following: $f(1)  = r$, and, for $1 \leq i  <2n-1$, 
\[
f(i+1)=\left\{\begin{array}{ll}
\mbox{a child of $f(i)$ that is not in $f(\{1,\dots,i-1\})$} & \mbox{if there is such a node};\\
\mbox{the parent of $f(i)$} & \mbox{otherwise.}
\end{array}\right.
\]
Note that $f$ is onto but not one-to-one.  In particular, each vertex $v$ of $T$, different from the root is mapped by $f^{-1}$ to $\deg_T(v)$ different vertices of the path $P =(1, 2,\dots, 2n-1)$ (see Figure~\ref{fig:tree1}(b)). The root~$r$ is mapped to $\deg_T(r)+1$ vertices of~$P$. Intuitively, each edge of~$T$ is  mapped by $f^{-1}$ on exactly two edges of~$P$. Note that if  $1\leq i,j \leq 2n-1$ are two integers such that $i$ is the least integer satisfying $f(i)  = v$, and $j$ is the largest  integer satisfying $f(j)  = v$,  then,  for every $i \leq k \leq j$,  $f(k)$ is either a descendent of $v$, or $v$ itself.  

\begin{figure}[htb]
\centering
\scalebox{0.6}{
\input{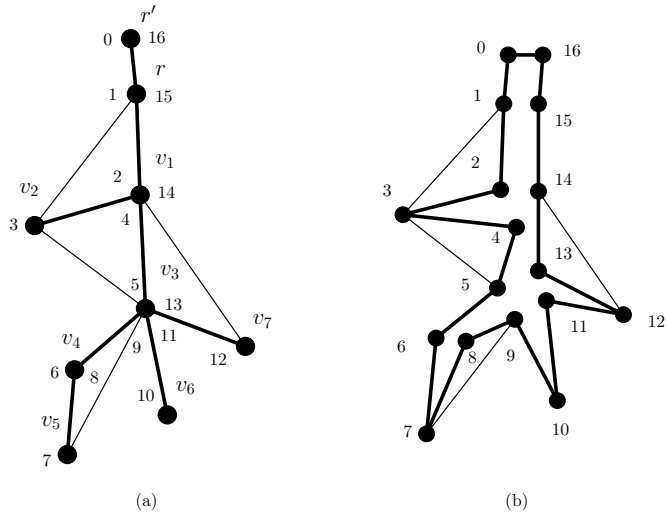}
}
\caption{Transformation of a planar graph to a path planar graphs.}\label{fig:tree1}
\end{figure}

Given $(G,T,f)$, we define a class of graphs induced by $(G,T,f)$ in such a way that: (1)~if $G$ is planar, then there exists a choice of $T$ and $f$ for which a graph induced by $(G,T,f)$ is path-outerplanar, and (2)~if $G$ is not planar, then for every $T$ and $f$, no graph induced by $(G,T,f)$ is path-outerplanar. Figure~\ref{fig:tree1} provides an illustration of the definition below (again, ignore node~$r'$ at this stage of the proof).

\begin{definition}\label{def:induced}
A  graph $H$  is  said to be \emph{induced} by $(G, T, f)$ if $V(H) =  \{1, \dots, 2n-1  \}$, and  (1)~$\{i,i+1\}\in E(H)$ for every $1\leq i < 2n-1$,  and (2)~for every cotree edge $\{u, v \}\in E(G)  \setminus E(T)$, there exists a unique edge $\{ i, j\} \in E(H)$ such that $\{f(i), f(j) \}  =  \{u, v \}$. 
\end{definition}

Note that the definition does not imply the uniqueness of $H$ as a graph induced by $(G,T,f)$. For example, in Figure~\ref{fig:tree1}(b), if we replace the edge between nodes $7$ and $9$ by an edge between $7$ and $5$, we get another graph induced by the graph on 
Figure~\ref{fig:tree1}(a) that satisfies Definition~\ref{def:induced}.

\begin{lemma}
\label{le:GTf}
For every planar graph $G$, and every spanning tree $T$ of $G$,  there exists a DFS-mapping $f$ of $T$, and a graph $G_{T,f}$ induced by $(G,T, f)$, such that $G_{T,f}$ is path-outerplanar.
\end{lemma}

\begin{proof}
Let $G=(V,E)$ be a planar graph, and consider a drawing of it in the plane.  Let $T$ be a spanning  tree  of $G$, rooted at an arbitrary node $r$. For every vertex $v\in V$,  we set a numbering 
\[
\nu_v:\{0,\dots,\deg_T(v)-1\}  \rightarrow N(v),
\] 
where $N(v)$ denotes the set of neighbors of~$v$ in $T$. In this numbering, $\nu_v(0)$ is the parent of $v$ in $T$ if $v\neq r$. These numberings satisfy that, for $i=0,\dots,\deg_T(v)-1$,  the edge $\{v, \nu_v(i)\}$ is immediately  followed by the edge $\{v, \nu_v(i+1 \bmod \deg_T(v))\}$ when one follows the counterclockwise ordering of the edges of $T$ incident to $v$ in the planar drawing of $G$.  For every two distinct children  $u, u'$ of $v$, the fact that      $u = \nu_v(k)$ and $u' = \nu_v(k')$ with $k <k'$ is denoted by $u \prec u'$.

We now define the DFS-mapping $f$ of $T$ as follows. The DFS traversal used for constructing $f$ starts from $r$, and explores the children of each node~$v$ in counterclockwise order, i.e., in the order provided by  the numbering~$\nu_v$. It follows that,  for every $v\in V$, there exists a sequence $1\leq i_1<\dots< i_d\leq 2n-1$ of integers, with $d=\deg_T(v)$ if $v\neq r$, and $d=\deg_T(r)+1$ otherwise, such that  $f^{-1}(v) = \{i_1,\dots, i_{d}\}$, and 
\[
f(i_1+1) = f(i_2-1) \prec f(i_2+1) = f(i_3-1) \prec \dots \prec f(i_{d-1}+1)=f(i_d-1).
\]
To construct the desired path-outerplanar graph  $G_{T,f}$, it is convenient to add an extra node~$r'$ to~$G$, of degree~1, connected to the root~$r$ of~$T$. In the drawing, $r'$~is placed so that $\{r,r'\}$ appears between the edges  $\{r,\nu_r(0)\}$ and $\{r,\nu_r(\deg_T(r)-1)\}$ when one follows the clockwise ordering of the edges of $T$ incident to $r$ in the planar drawing of $G$. Moreover, $r'$ is placed close enough to $r$ in the plane so that the edge $\{r,r'\}$ drawn as a straight line does not cross any edge in the drawing of $G$ (see Figure~\ref{fig:tree1}(a)). We extend $f$ to $\{0,2n\}$ by setting $f(0)=f(2n)=r'$. In this way, $r$ becomes an internal node of the tree $T'=T+r'$, avoiding special considerations in the construction of $G_{T,f}$  regarding whether a node~$v$ is the root~$r$ or not. In particular, $d$ systematically denotes the degree of the considered node $v\in V$, including~$r$, and thus  $f(i_1-1)=f(i_d+1)$ is the parent of~$v$ in~$T'$, for every $v\in V$. 

 For every $v\in V$, let us fix a circle $C_v$ centered at the point representing $v$ in the planar drawing of $G$. We choose $C_v$ sufficiently small so that (\emph{i})~$v$ is the only node inside $C_v$, (\emph{ii})~the only drawings of edges crossing $C_v$ have $v$ as endpoints, and (\emph{iii})~$C_v\cap C_{v'}=\emptyset$ whenever $v\neq v'$. See Figure~\ref{fig:thickening-a-tree}. For $1\leq k\leq d$, let $x_{i_k}$ be the point of the plane at the intersection  of $C_v$ with the drawing of the edge $\{ f(i_k),  f(i_k+1)) \}$. These points are met in the order $x_{i_1}, \dots, x_{i_d}$ whenever one travels along $C_v$ counterclockwise, starting from~$x_{i_1}$.  
 
 We say that a point $x$  of $C_v$ is of type $i_k$ if, traveling along  $C_v$ counterclockwise  starting from~$x$,   the first point met among $\{x_{i_1}, \dots, x_{i_d}\}$ is $x_{i_k}$. The type of $x$ is denoted by $\tau(x)$.   For every cotree edge $\{u, v\}\in E  \setminus E(T)$, the drawing of $\{u, v\}$ meets $C_u$ in a point  $x_{(u, v)}$, and meets $C_v $ in a point  $x_{(v, u)}$ (see Figure~\ref{fig:thickening-a-tree}).  We now consider the graph $G_{ T,f }$ induced by $(G, T, f)$, with the specific requirement that,  for every cotree edge $\{u, v \}$ of $E(G)  \setminus E(T)$, the unique edge $\{ i, j\} \in E(G_{ T,f })$ such that $\{f(i), f(j) \}  =  \{u, v \}$ is the edge 
 $
 \{ i, j\} = \{\tau(x_{(u, v)}), \tau(x_{(v, u)})   \}.
 $
 
 The proof of the lemma is completed by establishing that $G_{ T,f }$ is outerplanar. To do so, it is convenient to state the following simple result. For a graph $G=(V,E)$ with vertices numbered from $1$ to $n$, let $G^+$ be the graph obtained from $G$ by adding two vertices~$0$ and $n+1$, and $E(G^+)=E\;\cup \big \{\{0,1\}, \{0,n+1\}, \{n,n+1\}\big\}$. 

\begin{claim}\label{claim:po-embedding}
If $G$ is path-outerplanar with witness $(1,2,\dots, n)$, then $G^+$ is outerplanar, and has a drawing in which the outerface corresponds to the cycle $(0,1\dots,n,n+1)$. Conversely, if $G^+$ is outerplanar with a drawing in which the outerface forms a cycle $(0,1,\dots,n,n+1)$ in $G^+$, then $G$ is path-outerplanar with witness $(1,2,\dots,n)$. \hfill $\diamond$
\end{claim}

By Claim~\ref{claim:po-embedding}, it is sufficient to provide a planar embedding of $G_{T, f}^+$ in which the vertices  $\{0,\dots, 2n\}$, which form a cycle, are on the same face, in order. For this purpose, at each node~$v$ of $G$, the circle $C_v$ is split in $d$ sections, such that section~$k$, for $1\leq k\leq d$, contains all the points of $C_v$ with type~$i_k$ (see Figure~\ref{fig:thickening-a-tree}). We draw node $i_k$ on the middle of the section~$k$. In this way, all the nodes of $G_{T, f}$ are placed in the plane. 
Each edge $\{i,j\}$ of $G_{T, f}$ is drawn as a straight line connecting the corresponding points in the plane. In addition, the extra node~$r'$ is split into two nodes~$0$ and~$2n$, and three edges are added: $\{0,1\}, \{0,2n\}$, and $\{2n-1,n\}$. The two nodes ~$0$ and~$2n$ are placed next to each other, close enough so that these three edges do not intersect. 

The transformation simply corresponds to thickening the edges, and the internal nodes of the tree~$T'$, where the thickness of the tree is at most the maximum, taken over all the nodes~$v$, of the radius of~$C_v$, which can be chosen as small as desired. It follows that planarity is preserved. 

By the setting of~$f$, the nodes $(0,\dots,2n)$ form a path. By construction of $G_{T,f}$,  the cycle $(0,1,\dots,2n)$ in $G_{T,f}^+$ draws in a plane a Jordan curve, splitting the plane into two regions. The one containing the original drawing of $T$ is also a face of $G_{T,f}^+$, all the cotree edges being drawn on the other region.
It then follows from Claim~\ref{claim:po-embedding} that $G_{T,f}$ is path-outerplanar. 
\end{proof}

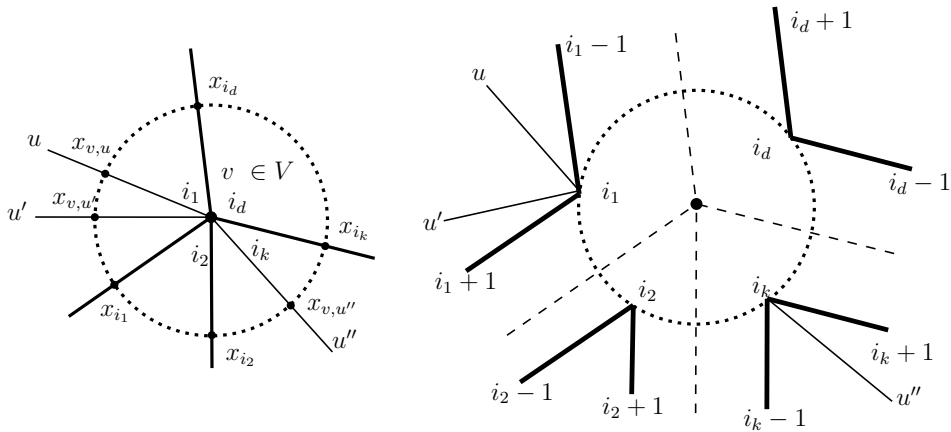
\begin{figure}[htb]
\begin{center}
\scalebox{0.8}{
\tikzset{every picture/.style={line width=0.75pt}} 

\begin{tikzpicture}[scale=0.75,x=0.75pt,y=0.75pt,yscale=-1,xscale=1]

\draw  [dash pattern={on 1.69pt off 2.76pt}][line width=1.5]  (123.51,313.5) .. controls (123.51,260.07) and (166.83,216.75) .. (220.26,216.75) .. controls (273.69,216.75) and (317.01,260.07) .. (317.01,313.5) .. controls (317.01,366.93) and (273.69,410.24) .. (220.26,410.24) .. controls (166.83,410.24) and (123.51,366.93) .. (123.51,313.5) -- cycle ;
\draw  [fill={rgb, 255:red, 0; green, 0; blue, 0 }  ,fill opacity=1 ] (217.55,310.79) .. controls (217.55,309.3) and (218.77,308.09) .. (220.26,308.09) .. controls (221.75,308.09) and (222.97,309.3) .. (222.97,310.79) .. controls (222.97,312.29) and (221.75,313.5) .. (220.26,313.5) .. controls (218.77,313.5) and (217.55,312.29) .. (217.55,310.79) -- cycle ;
\draw  [fill={rgb, 255:red, 0; green, 0; blue, 0 }  ,fill opacity=1 ] (206.32,217.56) .. controls (206.32,216.07) and (207.53,214.86) .. (209.03,214.86) .. controls (210.52,214.86) and (211.74,216.07) .. (211.74,217.56) .. controls (211.74,219.06) and (210.52,220.27) .. (209.03,220.27) .. controls (207.53,220.27) and (206.32,219.06) .. (206.32,217.56) -- cycle ;
\draw    (84.27,254.37) -- (220.26,310.79) ;

\draw [line width=1.5]    (101.86,394.01) -- (220.26,310.79) ;

\draw [line width=1.5]    (220.94,437.31) -- (220.26,310.79) ;

\draw [line width=0.75]    (321.07,423.78) -- (220.26,310.79) ;

\draw [line width=1.5]    (356.25,345.3) -- (220.26,310.79) ;

\draw [line width=1.5]    (203.08,169.39) -- (220.26,310.79) ;

\draw  [fill={rgb, 255:red, 0; green, 0; blue, 0 }  ,fill opacity=1 ] (216,310.79) .. controls (216,308.44) and (217.91,306.53) .. (220.26,306.53) .. controls (222.61,306.53) and (224.52,308.44) .. (224.52,310.79) .. controls (224.52,313.15) and (222.61,315.05) .. (220.26,315.05) .. controls (217.91,315.05) and (216,313.15) .. (216,310.79) -- cycle ;
\draw  [fill={rgb, 255:red, 0; green, 0; blue, 0 }  ,fill opacity=1 ] (129.47,274.39) .. controls (129.47,272.9) and (130.68,271.69) .. (132.17,271.69) .. controls (133.67,271.69) and (134.88,272.9) .. (134.88,274.39) .. controls (134.88,275.89) and (133.67,277.1) .. (132.17,277.1) .. controls (130.68,277.1) and (129.47,275.89) .. (129.47,274.39) -- cycle ;
\draw  [fill={rgb, 255:red, 0; green, 0; blue, 0 }  ,fill opacity=1 ] (120.81,310.79) .. controls (120.81,309.3) and (122.02,308.09) .. (123.51,308.09) .. controls (125.01,308.09) and (126.22,309.3) .. (126.22,310.79) .. controls (126.22,312.29) and (125.01,313.5) .. (123.51,313.5) .. controls (122.02,313.5) and (120.81,312.29) .. (120.81,310.79) -- cycle ;
\draw  [fill={rgb, 255:red, 0; green, 0; blue, 0 }  ,fill opacity=1 ] (137.04,367.49) .. controls (137.04,365.99) and (138.26,364.78) .. (139.75,364.78) .. controls (141.25,364.78) and (142.46,365.99) .. (142.46,367.49) .. controls (142.46,368.98) and (141.25,370.19) .. (139.75,370.19) .. controls (138.26,370.19) and (137.04,368.98) .. (137.04,367.49) -- cycle ;
\draw  [fill={rgb, 255:red, 0; green, 0; blue, 0 }  ,fill opacity=1 ] (218.1,409.7) .. controls (218.1,408.21) and (219.31,407) .. (220.8,407) .. controls (222.3,407) and (223.51,408.21) .. (223.51,409.7) .. controls (223.51,411.2) and (222.3,412.41) .. (220.8,412.41) .. controls (219.31,412.41) and (218.1,411.2) .. (218.1,409.7) -- cycle ;
\draw  [fill={rgb, 255:red, 0; green, 0; blue, 0 }  ,fill opacity=1 ] (283.72,384.81) .. controls (283.72,383.31) and (284.93,382.1) .. (286.43,382.1) .. controls (287.92,382.1) and (289.13,383.31) .. (289.13,384.81) .. controls (289.13,386.3) and (287.92,387.51) .. (286.43,387.51) .. controls (284.93,387.51) and (283.72,386.3) .. (283.72,384.81) -- cycle ;
\draw  [fill={rgb, 255:red, 0; green, 0; blue, 0 }  ,fill opacity=1 ] (312.41,335.01) .. controls (312.41,333.52) and (313.62,332.31) .. (315.11,332.31) .. controls (316.61,332.31) and (317.82,333.52) .. (317.82,335.01) .. controls (317.82,336.51) and (316.61,337.72) .. (315.11,337.72) .. controls (313.62,337.72) and (312.41,336.51) .. (312.41,335.01) -- cycle ;
\draw    (73.99,310.66) -- (220.26,310.79) ;

\draw  [dash pattern={on 1.69pt off 2.76pt}][line width=1.5]  (525.78,302.45) .. controls (525.78,248.31) and (569.67,204.43) .. (623.81,204.43) .. controls (677.94,204.43) and (721.83,248.31) .. (721.83,302.45) .. controls (721.83,356.59) and (677.94,400.47) .. (623.81,400.47) .. controls (569.67,400.47) and (525.78,356.59) .. (525.78,302.45) -- cycle ;
\draw  [fill={rgb, 255:red, 0; green, 0; blue, 0 }  ,fill opacity=1 ] (621.06,299.71) .. controls (621.06,298.19) and (622.29,296.97) .. (623.81,296.97) .. controls (625.32,296.97) and (626.55,298.19) .. (626.55,299.71) .. controls (626.55,301.22) and (625.32,302.45) .. (623.81,302.45) .. controls (622.29,302.45) and (621.06,301.22) .. (621.06,299.71) -- cycle ;
\draw    (449.28,200.31) -- (526.33,288.6) ;

\draw [line width=0.75]  [dash pattern={on 4.5pt off 4.5pt}]  (466.28,409.8) -- (623.81,299.71) ;

\draw [line width=0.75]  [dash pattern={on 4.5pt off 4.5pt}]  (623.12,475.06) -- (623.81,299.71) ;

\draw [line width=0.75]    (786.54,465.18) -- (682.62,379.09) ;

\draw [line width=0.75]  [dash pattern={on 4.5pt off 4.5pt}]  (788.19,341.25) -- (623.81,299.71) ;

\draw [line width=0.75]  [dash pattern={on 4.5pt off 4.5pt}]  (606.4,156.44) -- (623.81,299.71) ;

\draw  [fill={rgb, 255:red, 0; green, 0; blue, 0 }  ,fill opacity=1 ] (619.49,299.71) .. controls (619.49,297.32) and (621.42,295.39) .. (623.81,295.39) .. controls (626.19,295.39) and (628.13,297.32) .. (628.13,299.71) .. controls (628.13,302.09) and (626.19,304.03) .. (623.81,304.03) .. controls (621.42,304.03) and (619.49,302.09) .. (619.49,299.71) -- cycle ;
\draw    (413.64,313.83) -- (526.33,288.6) ;

\draw [line width=2.25]    (802.99,270.51) -- (702.36,244.18) ;

\draw [line width=2.25]    (688.93,134.51) -- (702.36,244.18) ;

\draw [line width=2.25]    (783.25,405.41) -- (682.62,379.09) ;

\draw [line width=2.25]    (682.35,471.76) -- (682.57,396.64) -- (682.62,379.09) ;

\draw [line width=2.25]    (569.93,459.15) -- (571.3,384.57) ;

\draw [line width=2.25]    (432.28,355.92) -- (526.33,291.34) ;

\draw [line width=2.25]    (477.25,449.14) -- (571.3,384.57) ;

\draw [line width=2.25]    (508.51,165.76) -- (526.33,291.34) ;

\draw (230.81,200.57) node  [font=\large]  {$x_{i_d}$};
\draw (120.99,252.18) node  [font=\large]  {$x_{v,u}$};
\draw (107.32,298.3) node  [font=\large]  {$x_{v,u'}$};
\draw (140.6,391.45) node  [font=\large]  {$x_{i_{1}}$};
\draw (245.14,428.53) node  [font=\large]  {$x_{i_{2}}$};
\draw (339.54,322.71) node  [font=\large]  {$x_{i_{k}}$};
\draw (321.2,387.08) node  [font=\large]  {$x_{v,u''}$};
\draw (203.73,290.46) node  [font=\large]  {$i_{1}$};
\draw (211.36,343.37) node  [font=\large]  {$i_{2}$};
\draw (261.71,338.18) node  [font=\large]  {$i_{k}$};
\draw (242.23,298.3) node  [font=\large]  {$i_{d}$};
\draw (73.18,246.52) node  [font=\large]  {$u$};
\draw (259.04,271.59) node  [font=\large]  {$v\ \in V$};
\draw (59.11,305.51) node  [font=\large]  {$u'$};
\draw (331.35,412.68) node  [font=\large]  {$u''$};
\draw (552.94,290.4) node  [font=\large]  {$i_{1}$};
\draw (583.43,372.23) node  [font=\large]  {$i_{2}$};
\draw (678.32,368.75) node  [font=\large]  {$i_{k}$};
\draw (678.33,253.04) node  [font=\large]  {$i_{d}$};
\draw (442.98,191.26) node  [font=\large]  {$u$};
\draw (404.59,314.1) node  [font=\large]  {$u'$};
\draw (800.53,460.52) node  [font=\large]  {$u''$};
\draw (432.01,365.65) node  [font=\large]  {$i_{1} +1$};
\draw (479.17,458.33) node  [font=\large]  {$i_{2} -1$};
\draw (571.3,470.39) node  [font=\large]  {$i_{2} +1$};
\draw (689.2,480.81) node  [font=\large]  {$i_{k} -1$};
\draw (795.43,424.91) node  [font=\large]  {$i_{k} +1$};
\draw (809.85,282.55) node  [font=\large]  {$i_{d} -1$};
\draw (728.33,146.75) node  [font=\large]  {$i_{d} +1$};
\draw (540.88,166.88) node  [font=\large]  {$i_{1} -1$};

\end{tikzpicture}
}
\end{center}
\vspace{-4ex}
\caption{Planar drawing: from $(G,T,f)$ to $G^+_{T,f}$.}\label{fig:thickening-a-tree}
\end{figure}

The next result is a form of reciprocal of Lemma~\ref{le:GTf}. 

\begin{lemma}
\label{le:noGTf}
If $G$ has a spanning tree $T$, and a DFS-mapping $f$ of $T$, such that a graph $H$  induced by  $(G, T, f) $  is path-outerplanar, then $G$ is planar. 
\end{lemma}

\begin{proof} 
Let us assume that $G$ has a spanning tree $T$, and a DFS-mapping $f$ of $T$ such that some graph $H$  induced by  $(G, T, f) $  is path-outerplanar. Let $G_f $ be the graph with vertex set $\{1, 2, \dots, 2n-1\}$ such that, for every $1\leq i<j\leq 2n-1$, $\{i,j\}$ is an edge if and only if (1)~$j=i+1$, or (2)~$f(i)  = f(j)$,  and every $k\in\{i+1,\dots,j-1\}$ satisfies $f(k) \neq  f(i)$. The graph $G_f$ is path-outerplanar by construction, due to the definition of the mapping $f$. Let $G'$ be the graph defined as $V(G') = \{1, \dots, 2n -1 \}$, and  $E(G')=E(H) \cup E(G_f)$.  $G'$ is planar, since both $H$ and $G_f$ are  path-outerplanar graphs, and thus, for each of the two graphs, its edges can be drawn in one of the two sides of the path $(1, \dots, 2n -1)$ --- if these two graphs share edges, then they are considered as belonging to $H$. It follows that if all edges of $G'$ belonging  to $E(G_f) \setminus E(H)$ are contracted, then the resulting graph is planar. Now, observe that the resulting graph is precisely the graph $G$, because the contractions identify all pairs of nodes $i,j$ for which $f(i) = f(j)$ holds. 
\end{proof}

\subsection{Certifying planar graphs}
\label{se:PLS-planar}

The previous sections provide us with the ingredients for the proof of Theorem~\ref{theo:main}. Lemma~\ref{le:GTf} provides a tool for transforming a planar graph into a path-outerplanar graph, by ``cutting along'' a spanning tree, and transforming the tree into a path through a specific DFS traversal. Importantly, this transformation is both ways, in the sense that, as stated in Lemma~\ref{le:noGTf}, if a graph $G$ has been transformed into a path-outerplanar graph by cutting along a spanning tree, then $G$ is necessarily planar. Our proof-labeling scheme for planar graphs aims at implementing this transformation, and then checking the path-outerplanarity of the resulting graph using the scheme provided in Lemma~\ref{lem:PLS-path-outerplanar}. The rest of the section contains the details of the implementation, including the delicate issue of keeping the certificate size small, which is not direct, as a same node of the graph may have to simulate the behavior of many vertices of the path-outerplanar graph after applying the transformation. 

\paragraph{Proof of Theorem~\ref{theo:main}.}

Given a planar graph $G$, the prover  draws it on the plane, constructs a spanning tree $T$, a DFS-mapping $f$ on its vertices, and a path-outerplanar graph $G_{T,f}$ as in Lemma~\ref{le:GTf}. The prover then provides the nodes with certificates allowing them to verify that:
\begin{enumerate}
\item\label{i:DFSo} $T$ is a spanning tree of $G$, and $f$ is a DFS-mapping of $T$; 
\item $G_{T,f}$ is path-outerplanar, with witness $f$.
\end{enumerate}
Certifying the conditions in the first item is quite simple and standard~\cite{KormanKP10}. It is therefore omitted, apart from the issue of the certificates size, which we discuss later. We focus on the second item. The proof-labeling scheme of the path-outerplanarity of $G_{T,f}$ with witness~$f$ was presented in Lemma~\ref{lem:PLS-path-outerplanar}.  Let us recall that a vertex $x$ of $G$ is transformed  into several nodes in $G_{T,f}$, namely into the nodes in $f^{-1}(x) = \{i_1,\dots,i_{d}\}$. The verification at node $x$ involves the simulation of the verification protocol for path-outerplanarity  for each of the nodes $i_1,\dots,i_{d}$ of the graph $G_{T,f}$. The index~$d$ may however be large, even as large as $O(n)$. It follows that, for keeping certificates of logarithmic size, the prover cannot simply assign  the $d$ path-outerplanarity certificates of its $d$ copies in $G_{T,f}$ to node $x$. We now explain how to cope with this issue for preserving $O(\log n)$-bit certificates.

Every edge $\{x,y\}$ of $T$ is mapped on two edges $\{i,j\}$ and $\{i',j'\}$ of $G_{T,f}$, and every non-tree edge $\{x,y\}$ of $G$ is mapped on one edge $\{i,j\}$ of $G_{T,f}$. The certificates corresponding to the $2n-1$ vertices of $G_{T,f}$ are distributed to the $n$ vertices of $G$ in such a way that, after a single communication round in $G$, each vertex~$x$ collects  the certificates of all its copies $i_1,\dots,i_{d}$ in the protocol for path-planarity, as well as the certificates of all the neighbors in $G_{T,f}$ of each copy~$i_j$, $j=1,\dots,d$. To see how this is achieved, let us consider for now a virtual scenario in which the certificates can be distributed on the \emph{edges} of $G$. Let $e=\{x,y\}$ be an edge of $G$, and let $\{i,j\}$ and $\{i',j'\}$ be its two associated edges in $G_{T,f}$ (if $e$ is mapped to a unique edge of $G_{T,f}$, which is the case of cotree edges, we simply set $\{i',j'\} = \{i,j\}$). Then the certificate~$c(e)$ of the edge~$e$ includes the following information:
\begin{itemize}
\item the identifiers of $x$ and $y$ in $G$;
\item the values $i,j,i',j'\in\{1,\dots,2n-1\}$;
\item the certificates of nodes $i,j,i',j'$ in the proof-labeling scheme for path-outerplanarity in $G_{T,f}$ with witness $f$.
\end{itemize}
If we could assign $c(e)$ to every edge $e=\{x,y\}$ of $G$, then $c(e)$ could be sent to its endpoints~$x$ and~$y$, and vertex $x$ could simulate the verification protocol of Algorithm~\ref{al:path-outerplanar} for checking the path-outerplanar graph $G_{T,f}$ with witness $f$, for each copy $i_1\dots,i_d$ of $x$. Since assigning certificates to edges is not doable, the certificate $c(e)$ is actually assigned to one of its two endpoints. This assignment is performed so that each node receives at most five edge-certificates. For this purpose we use the fact that a planar graph $G=(V,E)$ has \emph{degeneracy} at most~5. That is, there exists a total ordering~$\sigma$ of the vertices of $G$ such that every node $x$ has at most five neighbors $y$ with $\sigma(y)>\sigma(x)$. Every edge $e$ is associated to its smaller endpoint, according to $\sigma$. It follows that each node $x$ is associated to at most five edges, all with $x$ as endpoint. The prover assigns the certificate $c(x)$ to node $x$, defined as the concatenation of the at most five edge-certificates associated to~$x$. 
Extra-informations are also added to every certificate, for allowing the nodes to check that $T$ is a spanning tree of $G$. Overall,  the certificates are of the due size $O(\log n)$ bits.

The verification protocol executed at each node of $G$ is displayed in Algorithm~\ref{al:planar}.

\begin{figure}[h]
\begin{center}
\small
\begin{minipage}{.88\linewidth}
\begin{algorithm}[H]
\caption{\small Verification procedure for planarity at node $x\in V(G)$}
\label{al:planar}
collect the certificate of each neighbor; \\
\vspace{1ex}
//\emph{\textbf{Phase 1:} recover local information regarding $T,f$, and $G_{T,f}$;}//\\
compute ID of the parent of $x$ in $T$, and the IDs of all its children in $T$;\\
compute $f^{-1}(x)=\{i_1,i_2,\dots,i_d\}$, i.e., all nodes corresponding to $x$ in $G_{T,f}$;\\
\For{{\rm every} $i \in f^{-1}(x)$}
{
	retrieve certificate of node $i$ in $G_{T,f}$  dedicated to the path-outerplanarity of $G_{T,f}$; \\
	compute the neighbors of $i$ in $G_{T,f}$, and their certificates for path-outerplanarity;\\
}
\vspace{1ex}
//\emph{\textbf{Phase 2:} check that $T$ is a spanning tree of $G$, and $f$ is a DFS-mapping of $T$;}//\\
check local consistency of certificates for spanning tree;\label{t:spanningT}\\
check local consistency of certificates for DFS-mapping;\label{t:DFSo}\\

\vspace{1ex}

// \emph{\textbf{Phase 3:} check the path-outerplanarity of $G_{T,f}$ with witness $f$;}//\\
\For{{\rm every} $i \in f^{-1}(x)$}
{
	simulate execution of Algorithm~\ref{al:path-outerplanar} at node $i$ of $G_{T,f}$;\label{t:po}\\
}
\vspace{1ex}

\textbf{if} all checks are passed \textbf{then} return accept \textbf{else} return reject.

\end{algorithm}
\end{minipage}
\end{center}
\end{figure}

In Phase~1 of Algorithm~\ref{al:planar}, every node $x$ retrieves information regarding $x$ itself as a node of $G$, and all the ``virtual'' nodes $i\in f^{-1}(x)$ of $G_{T,f}$. Observe that, for every edge $e=\{x,y\}$ incident to $x$ in $G$, node $x$ obtains all information contained in $c(e)$ after one single communication round with its neighbors. An edge $e=\{x,y\}$ is in $T$ if and only if the pair $(x,y)$ has been mapped on two disjoint pairs $(i,j)$, $(i',j')$ corresponding to distinct edges $G_{T,f}$. Node $y$ is the parent of $x$ if $\min\{i,i'\}<\min\{j,j'\}$, and is a child otherwise. By considering all  incident edges in $G$, $x$ retrieves all its corresponding nodes $i_1,\dots, i_d$ in $G_{T,f}$, their neighbors in $G_{T,f}$, and all the corresponding  certificates for the path-outerplanarity of $G_{T,f}$ with witness $f$.

Phase 2 checks that $T$ is a spanning tree of  $G$ (see~\cite{KormanKP10}), and checks that $f$ is a DFS-mapping of $T$. This latter condition is rather simple to check, so we only provide some hints. Let $f^{-1}_{\min}(x)$ and $f^{-1}_{\max}(x)$ be the smallest and largest values in $f^{-1}(x)$. These values correspond to the first, and last time the DFS visited node $x$. The children $y_1, \dots, y_k$ of $x$ in $T$ are then sorted by increasing value of $f^{-1}_{\min}(y_i)$. The main test consists to check that: 
\begin{itemize}
\item if $x$ has no children, then $f^{-1}_{\max}(x) = f^{-1}_{\min}(x)$; 
\item if $x$ has children $y_1, \dots, y_k$, then $f^{-1}_{\min}(x) = f^{-1}_{\min}(y_1)-1$, $f^{-1}_{\max}(x) = f^{-1}_{\max}(y_k)+1$, and, for every $j=1,\dots, k$, $f^{-1}_{\min}(y_{j+1})=f^{-1}_{\max}(y_j)+2$. 
\end{itemize}
In addition, the root $r$ of $T$ checks that $\{1,2n-1\}\subseteq f^{-1}(r)$.

Finally, Phase~3 of Algorithm~\ref{al:planar} at $x$ applies the verification protocol for checking the consistency of the distributed proof of the path-outerplanarity of $G_{T,f}$, for every node $i\in f^{-1}(x)$  of $G_{T,f}$ corresponding to~$x$. 

Completeness holds as, if the input graph $G$ is planar, and if the prover constructs $T$ and $f$ correctly, as well as   the path-outerplanar graph $G_{T,f}$ whose existence is guaranteed by Lemma~\ref{le:GTf}, then Phase 2 accepts by construction of $T$ and~$f$, and Phase 3 accepts since the  verification algorithm for path-outerplanarity is correct by Lemma~\ref{lem:PLS-path-outerplanar}.

For the soundness, assume that Algorithm~\ref{al:planar} returns accepts at all nodes of $G$. Then, thanks to Phase~2,   $T$ is necessarily a spanning tree of~$G$, and $f$ is a DFS-mapping of $T$. Also, thanks to Phase~3, $G_{T,f}$ is necessarily path-outerplanar with witness $f$. Therefore, by Lemma~\ref{le:noGTf},  $G$ is planar. This completes the proof of Theorem~\ref{theo:main}. \qed

\section{Lower bounds}

This section is dedicated to the proof of Theorem~\ref{theo:lwb}. The proof is modular: Section~\ref{subsec:Kk} treats the case of $\Forb(K_k)$, while  Section~\ref{subsec:Kpq} treats the case of $\Forb(K_{p,q})$. Finally, Section~\ref{subsec:lower-bound-proof} combines the two cases, and completes the proof. All proofs in this section can be generalized to the setting where the verification algorithm performs $t$ rounds instead of just one round, for any constant $t\geq 1$, by replacing some edges of the graph instances constructed hereafter by a path of length~$O(t)$.

\subsection{Lower bound for $\Forb(K_{k})$}
\label{subsec:Kk}

Our lower bound proof for $\Forb(K_{k})$ is based on extending the techniques developed in~\cite{FeuilloleyH18}. 

\begin{lemma}
\label{lem:Kk}
For every $k\geq 3$, there are no locally checkable proofs for $\Forb(K_k)$ with $o(\log n)$-bit certificates.
\end{lemma}

\begin{proof}
We first define two similar types of instances that we call \emph{paths of blocks} and \emph{cycles of blocks}. We prove that paths of blocks are legal instances, that is, they are $K_k$-minor-free, and that \emph{cycles of blocks} are illegal instances, that is, they contain $K_k$ as a minor. Then, we show by a counting argument that, if all the paths of blocks are accepted, and if the certificates used for these instances are too small, then there exists a cycle of blocks that is also accepted, which implies that the scheme is not correct.  

\medskip\noindent{\it Blocks and block connections.}
Let $p\geq 1$, and let $n=(k-1)(p+2)$.
For every $r=0,\dots,p+1$, let us consider a path $P_r$ of $k-1$ nodes with identifiers $r(k-1),r(k-1)+1,..., (r+1)(k-1)-1$, in consecutive order. 
A \emph{block} $B_r$ is a graph formed of such a path $P_r$, plus edges between every pair of nodes in~$P_r$. 
In other words, a block is a complete graph $K_{k-1}$, with an ordering of the nodes given by their identifiers. 
There are two special blocks: $B_0$, called the \emph{starting block}, and $B_{p+1}$, called the \emph{ending block}. The $p$ other blocks are called \emph{ordinary blocks}.
Note that the identifiers of the blocks are all distinct.

Blocks are connected by \emph{block connections}. A block connection from  block $B_i$ to block $B_j$ is as follows. Place the block $B_i$  on an horizontal line with the nodes in increasing order of their IDs, and then place the block $B_j$ to the right of $B_i$, also with the nodes in increasing order of their IDs. 
Then, add all edges between the $\lceil (k-1)/2 \rceil$ rightmost nodes of $B_i$  (which is called the \emph{right part} of $B_i$) and the $\lfloor (k-1)/2 \rfloor$ leftmost nodes of $B_j$ (which is called the \emph{left  part} of $B_j$). See Fig.~\ref{fig:block-connexion} for an example with $k=4$, where the edges inside the blocks are grey, and the edges between the blocks are black (since $\lceil (k-1)/2 \rceil=2$ and $\lfloor (k-1)/2 \rfloor=1$, there are two edges between these blocks).

\begin{figure}[htb]
\begin{center}
\scalebox{0.8}{
\tikzset{every picture/.style={line width=0.75pt}} 

\begin{tikzpicture}[x=0.75pt,y=0.75pt,yscale=-1,xscale=1]

\draw [color={rgb, 255:red, 155; green, 155; blue, 155 }  ,draw opacity=1 ][line width=1.5]    (100.5,169.5) .. controls (132.95,120.65) and (188.95,120.65) .. (219.5,169.5) ;

\draw [color={rgb, 255:red, 155; green, 155; blue, 155 }  ,draw opacity=1 ][line width=1.5]    (100.5,169.5) -- (159.5,169.5) ;

\draw [color={rgb, 255:red, 155; green, 155; blue, 155 }  ,draw opacity=1 ][fill={rgb, 255:red, 155; green, 155; blue, 155 }  ,fill opacity=1 ][line width=1.5]    (159.5,169.5) -- (219.5,169.5) ;

\draw  [fill={rgb, 255:red, 0; green, 0; blue, 0 }  ,fill opacity=1 ] (94.01,169.18) .. controls (94.18,165.6) and (97.23,162.83) .. (100.82,163.01) .. controls (104.4,163.18) and (107.17,166.23) .. (106.99,169.82) .. controls (106.82,173.4) and (103.77,176.17) .. (100.18,175.99) .. controls (96.6,175.82) and (93.83,172.77) .. (94.01,169.18) -- cycle ;
\draw  [fill={rgb, 255:red, 0; green, 0; blue, 0 }  ,fill opacity=1 ] (153.01,169.18) .. controls (153.18,165.6) and (156.23,162.83) .. (159.82,163.01) .. controls (163.4,163.18) and (166.17,166.23) .. (165.99,169.82) .. controls (165.82,173.4) and (162.77,176.17) .. (159.18,175.99) .. controls (155.6,175.82) and (152.83,172.77) .. (153.01,169.18) -- cycle ;
\draw  [fill={rgb, 255:red, 0; green, 0; blue, 0 }  ,fill opacity=1 ] (213.01,169.18) .. controls (213.18,165.6) and (216.23,162.83) .. (219.82,163.01) .. controls (223.4,163.18) and (226.17,166.23) .. (225.99,169.82) .. controls (225.82,173.4) and (222.77,176.17) .. (219.18,175.99) .. controls (215.6,175.82) and (212.83,172.77) .. (213.01,169.18) -- cycle ;
\draw [color={rgb, 255:red, 155; green, 155; blue, 155 }  ,draw opacity=1 ][line width=1.5]    (290.5,169.5) .. controls (322.95,120.65) and (378.95,120.65) .. (409.5,169.5) ;

\draw [color={rgb, 255:red, 155; green, 155; blue, 155 }  ,draw opacity=1 ][line width=1.5]    (290.5,169.5) -- (349.5,169.5) ;

\draw [color={rgb, 255:red, 155; green, 155; blue, 155 }  ,draw opacity=1 ][fill={rgb, 255:red, 155; green, 155; blue, 155 }  ,fill opacity=1 ][line width=1.5]    (349.5,169.5) -- (409.5,169.5) ;

\draw  [fill={rgb, 255:red, 0; green, 0; blue, 0 }  ,fill opacity=1 ] (284.01,169.18) .. controls (284.18,165.6) and (287.23,162.83) .. (290.82,163.01) .. controls (294.4,163.18) and (297.17,166.23) .. (296.99,169.82) .. controls (296.82,173.4) and (293.77,176.17) .. (290.18,175.99) .. controls (286.6,175.82) and (283.83,172.77) .. (284.01,169.18) -- cycle ;
\draw  [fill={rgb, 255:red, 0; green, 0; blue, 0 }  ,fill opacity=1 ] (343.01,169.18) .. controls (343.18,165.6) and (346.23,162.83) .. (349.82,163.01) .. controls (353.4,163.18) and (356.17,166.23) .. (355.99,169.82) .. controls (355.82,173.4) and (352.77,176.17) .. (349.18,175.99) .. controls (345.6,175.82) and (342.83,172.77) .. (343.01,169.18) -- cycle ;
\draw  [fill={rgb, 255:red, 0; green, 0; blue, 0 }  ,fill opacity=1 ] (403.01,169.18) .. controls (403.18,165.6) and (406.23,162.83) .. (409.82,163.01) .. controls (413.4,163.18) and (416.17,166.23) .. (415.99,169.82) .. controls (415.82,173.4) and (412.77,176.17) .. (409.18,175.99) .. controls (405.6,175.82) and (402.83,172.77) .. (403.01,169.18) -- cycle ;
\draw [line width=1.5]    (219.5,169.5) .. controls (237.95,154.65) and (273.95,155.65) .. (290.5,169.5) ;

\draw [line width=1.5]    (159.5,169.5) .. controls (199.5,139.5) and (248.95,133.65) .. (290.5,169.5) ;

\draw  [color={rgb, 255:red, 155; green, 155; blue, 155 }  ,draw opacity=1 ] (80.95,120.65) -- (239.95,120.65) -- (239.95,190.65) -- (80.95,190.65) -- cycle ;
\draw  [color={rgb, 255:red, 155; green, 155; blue, 155 }  ,draw opacity=1 ] (270.95,120.65) -- (429.95,120.65) -- (429.95,190.65) -- (270.95,190.65) -- cycle ;

\draw (159,210.55) node   [align=left] {$B_{i}$};
\draw (350,210.55) node   [align=left] {$B_{j}$};

\end{tikzpicture}
}
\vspace{-2ex}
\caption{Block connexion for $k=4$, from a block $B_i$ to a block $B_j$.}
\label{fig:block-connexion}
\end{center}
\end{figure}
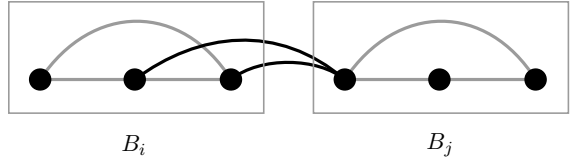

Note that two nodes that are linked by an edge have at most $k-3$ nodes between them in the horizontal ordering. In other words, no edge can ``jump'' over more than $k-3$ nodes.
Also note that, if $B_i$ has a block connection to $B_j$, and if $B_j$ has a block connection to $B_k$, then, every node of $B_j$ is connected either to a node of $B_i$, or to a node of $B_k$, but not both.  

\medskip\noindent{\it Path and cycles of blocks.}
Let $\pi$ be a permutation on $p$ elements. To create a path of blocks, $B_0$ has a block connection to $B_{\pi^{-1}(1)}$, $B_{\pi^{-1}(1)}$ has a block connection to $B_{\pi^{-1}(2)}$, etc., and $B_{\pi^{-1}(p)}$ has a block connection to $B_{p+1}$.
A cycle of blocks is constructed in a similar way, as follows. Let $k,  k'$ with $1 \leq k < k' \leq p$. To create a  cycle of blocks, $B_{\pi^{-1}(\ell)}$ has a connection to $B_{\pi^{-1}(\ell +1)}$ for every $\ell\in\{k ,\dots, k'-1\}$,   and  $B_{\pi^{-1}( k' )}$ has a connection to $B_{\pi^{-1}(k)}$.  Note that a cycle of blocks uses only a subset of blocks. 
Note also that, as the node identifiers given to different blocks are distinct, the paths and the cycles of blocks are  well defined.

\begin{claim}
Paths of blocks are $K_k$-minor-free.
\end{claim} 

\begin{proofofclaim} 
By symmetry, there is no loss of generality in proving the claim  only for the identity permutation~ $\pi$. For simplicity of the notation, the nodes are referred to by their identifiers.
Also, for the sake of clarity, we do not consider a path of blocks but an extension of it. That is, we add some edges, to make the path of blocks more symmetric, and hence avoid some case analysis.  
Note that, to prove the claim, it is sufficient to show that the larger graph $G$ obtained by adding edges to the path of blocks, is $K_k$-minor-free.  

For every $0 \leq r \leq p$, and every $0 \leq i \leq k-1$, the node $r(k-1) +i  $ of  $B_r$  is connected in $G$ to all nodes $(r+1)(k-1) +j$, $0 \leq j < i$,  of  $B_{r+1}$. See Figure~\ref{fig:kk-minor-free} for an example. 
Note that, for every node $\ell$, the closed neighborhood of $\ell$ (i.e., $\ell$ and its neighbors) is the interval 
 $[ \ell -k+2, \ell +k -2] $. It follows that, for every edge  $\{\ell, \ell'\}$, $\vert \ell - \ell'  \vert \leq k-2$. 

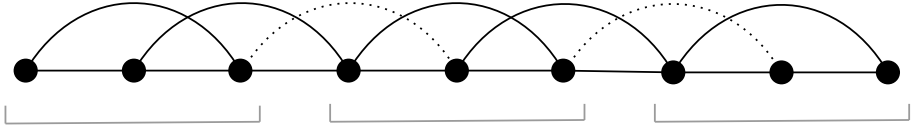
\begin{figure}[h]
\begin{center}
\scalebox{0.9}{
\tikzset{every picture/.style={line width=0.75pt}} 

\begin{tikzpicture}[x=0.75pt,y=0.75pt,yscale=-1,xscale=1]

\draw  [fill={rgb, 255:red, 0; green, 0; blue, 0 }  ,fill opacity=1 ][line width=0.75]  (54.6,201.2) .. controls (54.6,197.78) and (57.38,195) .. (60.8,195) .. controls (64.22,195) and (67,197.78) .. (67,201.2) .. controls (67,204.62) and (64.22,207.4) .. (60.8,207.4) .. controls (57.38,207.4) and (54.6,204.62) .. (54.6,201.2) -- cycle ;
\draw  [fill={rgb, 255:red, 0; green, 0; blue, 0 }  ,fill opacity=1 ][line width=0.75]  (114.6,201.2) .. controls (114.6,197.78) and (117.38,195) .. (120.8,195) .. controls (124.22,195) and (127,197.78) .. (127,201.2) .. controls (127,204.62) and (124.22,207.4) .. (120.8,207.4) .. controls (117.38,207.4) and (114.6,204.62) .. (114.6,201.2) -- cycle ;
\draw  [fill={rgb, 255:red, 0; green, 0; blue, 0 }  ,fill opacity=1 ][line width=0.75]  (173.6,201.2) .. controls (173.6,197.78) and (176.38,195) .. (179.8,195) .. controls (183.22,195) and (186,197.78) .. (186,201.2) .. controls (186,204.62) and (183.22,207.4) .. (179.8,207.4) .. controls (176.38,207.4) and (173.6,204.62) .. (173.6,201.2) -- cycle ;
\draw [line width=0.75]    (60.8,201.2) -- (120.8,201.2) ;

\draw [line width=0.75]    (120.8,201.2) -- (180.8,201.2) ;

\draw [line width=0.75]    (60.8,201.2) .. controls (90.43,151.4) and (150.43,150.4) .. (179.8,201.2) ;

\draw  [fill={rgb, 255:red, 0; green, 0; blue, 0 }  ,fill opacity=1 ][line width=0.75]  (233.6,201.2) .. controls (233.6,197.78) and (236.38,195) .. (239.8,195) .. controls (243.22,195) and (246,197.78) .. (246,201.2) .. controls (246,204.62) and (243.22,207.4) .. (239.8,207.4) .. controls (236.38,207.4) and (233.6,204.62) .. (233.6,201.2) -- cycle ;
\draw  [fill={rgb, 255:red, 0; green, 0; blue, 0 }  ,fill opacity=1 ][line width=0.75]  (293.6,201.2) .. controls (293.6,197.78) and (296.38,195) .. (299.8,195) .. controls (303.22,195) and (306,197.78) .. (306,201.2) .. controls (306,204.62) and (303.22,207.4) .. (299.8,207.4) .. controls (296.38,207.4) and (293.6,204.62) .. (293.6,201.2) -- cycle ;
\draw  [fill={rgb, 255:red, 0; green, 0; blue, 0 }  ,fill opacity=1 ][line width=0.75]  (352.6,201.2) .. controls (352.6,197.78) and (355.38,195) .. (358.8,195) .. controls (362.22,195) and (365,197.78) .. (365,201.2) .. controls (365,204.62) and (362.22,207.4) .. (358.8,207.4) .. controls (355.38,207.4) and (352.6,204.62) .. (352.6,201.2) -- cycle ;
\draw [line width=0.75]    (239.8,201.2) -- (299.8,201.2) ;

\draw [line width=0.75]    (299.8,201.2) -- (359.8,201.2) ;

\draw [line width=0.75]    (239.8,201.2) .. controls (269.43,151.4) and (329.43,150.4) .. (358.8,201.2) ;

\draw  [fill={rgb, 255:red, 0; green, 0; blue, 0 }  ,fill opacity=1 ][line width=0.75]  (413.6,202.2) .. controls (413.6,198.78) and (416.38,196) .. (419.8,196) .. controls (423.22,196) and (426,198.78) .. (426,202.2) .. controls (426,205.62) and (423.22,208.4) .. (419.8,208.4) .. controls (416.38,208.4) and (413.6,205.62) .. (413.6,202.2) -- cycle ;
\draw  [fill={rgb, 255:red, 0; green, 0; blue, 0 }  ,fill opacity=1 ][line width=0.75]  (473.6,202.2) .. controls (473.6,198.78) and (476.38,196) .. (479.8,196) .. controls (483.22,196) and (486,198.78) .. (486,202.2) .. controls (486,205.62) and (483.22,208.4) .. (479.8,208.4) .. controls (476.38,208.4) and (473.6,205.62) .. (473.6,202.2) -- cycle ;
\draw  [fill={rgb, 255:red, 0; green, 0; blue, 0 }  ,fill opacity=1 ][line width=0.75]  (532.6,202.2) .. controls (532.6,198.78) and (535.38,196) .. (538.8,196) .. controls (542.22,196) and (545,198.78) .. (545,202.2) .. controls (545,205.62) and (542.22,208.4) .. (538.8,208.4) .. controls (535.38,208.4) and (532.6,205.62) .. (532.6,202.2) -- cycle ;
\draw [line width=0.75]    (419.8,202.2) -- (479.8,202.2) ;

\draw [line width=0.75]    (479.8,202.2) -- (539.8,202.2) ;

\draw [line width=0.75]    (419.8,202.2) .. controls (449.43,152.4) and (509.43,151.4) .. (538.8,202.2) ;

\draw [line width=0.75]    (120.8,201.2) .. controls (150.43,151.4) and (210.43,150.4) .. (239.8,201.2) ;

\draw [line width=0.75]  [dash pattern={on 0.84pt off 2.51pt}]  (180.8,201.2) .. controls (210.43,151.4) and (270.43,150.4) .. (299.8,201.2) ;

\draw [line width=0.75]  [dash pattern={on 0.84pt off 2.51pt}]  (359.8,201.2) .. controls (389.43,151.4) and (450.43,151.4) .. (479.8,202.2) ;

\draw [line width=0.75]    (299.8,201.2) .. controls (329.43,151.4) and (390.43,151.4) .. (419.8,202.2) ;

\draw [line width=0.75]    (179.8,201.2) -- (239.8,201.2) ;

\draw [line width=0.75]    (358.8,201.2) -- (419.8,202.2) ;

\draw [color={rgb, 255:red, 155; green, 155; blue, 155 }  ,draw opacity=1 ]   (49.58,230.8) -- (190.58,229.8) ;

\draw [color={rgb, 255:red, 155; green, 155; blue, 155 }  ,draw opacity=1 ]   (49.58,220.8) -- (49.58,230.8) ;

\draw [color={rgb, 255:red, 155; green, 155; blue, 155 }  ,draw opacity=1 ]   (190.58,220.8) -- (190.58,229.8) ;

\draw [color={rgb, 255:red, 155; green, 155; blue, 155 }  ,draw opacity=1 ]   (229.58,229.8) -- (370.58,228.8) ;

\draw [color={rgb, 255:red, 155; green, 155; blue, 155 }  ,draw opacity=1 ]   (229.58,219.8) -- (229.58,229.8) ;

\draw [color={rgb, 255:red, 155; green, 155; blue, 155 }  ,draw opacity=1 ]   (370.58,219.8) -- (370.58,228.8) ;

\draw [color={rgb, 255:red, 155; green, 155; blue, 155 }  ,draw opacity=1 ]   (409.58,229.8) -- (550.58,228.8) ;

\draw [color={rgb, 255:red, 155; green, 155; blue, 155 }  ,draw opacity=1 ]   (409.58,219.8) -- (409.58,229.8) ;

\draw [color={rgb, 255:red, 155; green, 155; blue, 155 }  ,draw opacity=1 ]   (550.58,219.8) -- (550.58,228.8) ;

\end{tikzpicture}
}
\end{center}
\caption{\label{fig:kk-minor-free}
The graph $G$ that extends the path of blocks with three blocks, for $k=4$. The edges of the original path of blocks are plain edges and the new edges are dotted edges.}
\end{figure}

Let us assume, for the purpose of contradiction, that $G$ contains $K_k$ as minor. Then, there are $k$ disjoint sets of nodes $S_1,\dots,S_k$ such that each set $S_i$  is connected, and, for every pair $(S_i,S_j)$ of sets,  $S_i \cup S_j $ is a connected subgraph of~$G$.
Each node belongs to at most one set~$S_i$. 
Let us visit the nodes of the path of blocks, by increasing order of their IDs, and 
let $S_a$ be the last visited set. 
Let $\ell_a$ be the first node belonging to $S_a$, i.e., the node with minimum ID in $S_a$. 
Let us now consider a set $S_b$ such that  $S_b\cap [\ell_a- k+2,  \ell_a]=\emptyset$. Such a set exists because there are $k$ disjoint sets, and the interval is of size $k-1$. Let $\ell_b$ be a node of $S_b$ such that  $\ell_b < \ell_a$. Such a node exists because, by definition of $S_a$, $S_b$ is visited before $S_a$ when nodes are traversed according to the increasing order of their IDs.
The interval  $I =  [\ell_a- k+2,  \ell_a -1]$ contains   $k-2$ integers, and $(S_a  \cup S_b ) \cap I$ is empty. 
Now, on the one hand, there cannot exist an edge  $\{\ell, \ell'\}$ with $\ell < \ell_a- k+2$, and
  $\ell' > \ell_a -1$, because $\vert \ell - \ell'  \vert \geq k-1$. However, on the other hand, the sets 
  \[
  S_a  \cap [\ell_a, (r+1)(p+2) ] \;\;\mbox{and}\;\; S_b  \cap [0, \ell_a -k+1]
  \]
are both non-empty, as they respectively contain $\ell_a$  and $\ell_b$.  It follows that  $S_a  \cup S_b$  is not connected, and therefore $K_k$ cannot be a minor of the path of blocks.
\end{proofofclaim} 

\begin{claim}
Cycles of blocks are \emph{not} $K_k$-minor-free.
\end{claim}

\begin{proofofclaim} 
Let us consider a cycle of blocks, and let $B$ be an arbitrary block of this cycle. By construction, the graph obtained by removing $B$ from the cycle is connected. Let us contract this graph into one node~$v$. The node~$v$ is connected to all nodes of $B$ since, as pointed out before, every node of $B$ is linked to at least one node outside~$B$. Now, as $B$ is in itself $K_{k-1}$, the resulting contracted graph is~$K_k$. 
\end{proofofclaim} 

To complete the proof of Lemma~\ref{lem:Kk}, it is sufficient to show that cycles and paths of blocks are indistinguishable whenever the certificates are too small.

Let us assume that there exists a locally checkable proof with certificates of size $g(n)=o(\log n)$ bits. 
We show that there exist two paths of blocks for which the prover gives to every block the exact same certificates. For this purpose, 
let us call \emph{labeled block} a block in which every node is given a certificate of size $g(n)$. Every block can be labeled in $2^{(k-1)g(n)}$ manners. Let us consider a set of labeled blocks where one labeled version of each block appears as element of this set. 
There are $2^{(k-1)g(n)p}$ such sets. 
On the other hand, there are $p!$ different paths of blocks, as there are $p!$ permutations of the ordinary blocks. 
To compare these numbers, it is simpler to compare their logarithms. We get the following asymptotics (having in mind  that $k$ is a constant, and  $p= \Theta(n)$):
\[
\log\left(2^{(k-1)g(n)p}\right) = o(n\log n) 
\text{ and } 
\log(p!) = \Omega(n\log n).
\]
It follows that, for large $n$, there are more paths of blocks than distinct sets of labeled blocks. Thus, by the pigeon-hole principle, for large enough $n$, there exist two paths of blocks, $P$ and~$P'$, where all nodes accept with a same certificate assignment, i.e., certificates inducing identical labeled blocks. Let $A$ be this certificate assignment. 

Without loss of generality, let us assume that $P$ corresponds to the identity permutation. In the permutation $\pi$ corresponding to $P'$, there must exist two indices $i$ and $j$ with $i<j$ such that $\pi(i)>\pi(j)$. 
We define a cycle of blocks $C$ as follows. The blocks $B_i, B_{i+1},...,B_{j}$ are connected in this order, and the cycle is closed by connecting $B_{j}$ to $B_{i}$. 
We claim that every node in $C$  accepts whenever it is given certificates according to $A$. Indeed, let $u$ be any node of $C$, except the ones incident to an edge between  $B_{j}$ and~$B_{i}$. Node~$u$ has the same view in $C$ as in $P$ (i.e., it has the same neighbors, with the same identifiers, and the same certificates),  thus it accepts. Now, let us consider a node that is incident to an edge between  $B_{j}$ and~$B_{i}$. As pointed out before, such a node is not linked to any node from a third block. It follows that it has the same view in $C$ as in $P'$.  
Therefore, all the nodes of $C$ accept, in contradiction with the soundness condition. 
This completes the proof of Lemma~\ref{lem:Kk}.
\end{proof}

Note that the proof of Lemma~\ref{lem:Kk} can be adapted to a larger verification radius~$t$, by replacing each edge by a path of length~$t$. 

\subsection{Lower bound for $\Forb(K_{p,q})$}
\label{subsec:Kpq}

Our lower bound proof for $\Forb(K_{p,q})$ is an adaptation of the lower bound proof for spanning tree and leader election, by G\"o\"os and Suomela~\cite{GoosS16}. Roughly speaking, we show that, if the certificates are of size $o(\log n)$, then it is possible to define a set of legal instances (i.e.,  $K_{p,q}$-free graphs) which can be ``glued''  together in order to obtain an illegal instance (i.e., a graph containing $K_{p,q}$ as a minor) in such a way that no vertices can locally distinguish which instance they belong to. 

\begin{lemma}
\label{lem:Kpq}
For every $p,q \geq 2$, there are no locally checkable proofs for $\Forb(\{ K_{p,q} \})$ using certificates on $o(\log n)$ bits.
\end{lemma}

\begin{proof}
 If $p=q=2$, then $K_{p,q}$ is a 4-cycle, and thus $\Forb(\{ K_{p,q} \})$ is the family of graphs in which the only allowed cycles are of length~3. The fact that this class of graphs cannot be certified by a locally checkable proof using certificates on $o(\log n)$ bits follows directly for previously known constructions (see \cite{GoosS16}).

Let $p \geq 2$ and $q\geq 3$, and let us assume, for the purpose of contradiction, that $\Forb(\{ K_{p,q} \})$ admits a locally checkable proof with certificates of size $o(\log n)$. Without loss of generality, we assume that $q\geq p$, and we consider  $n \geq 6q$. The range of identifiers is split into subsets, as in~\cite{GoosS16}. Let $n_A = \lfloor n/2 \rfloor$, and $n_B = \lceil n/2 \rceil$. Let us fix any partition of the identifiers $\{1, \dots, n^2\}$ into $2n$ subsets $a_1, \dots, a_n, b_1, \dots, b_n$ so that each $a_i \in A= \{a_1, \dots, a_n\}$ is of size $n_A$, and each $b_i\in  B = \{b_1, \dots, b_n\}$ is of size $n_B$. We denote by $a_i[j]$ (resp. $b_i[j]$) the $j$-th identifier of $a_i\in A$ (resp.~$b_i\in B$) in increasing order of the IDs in $a_i$ (resp.~$b_i$).

For every pair $(a,b)\in A\times B$, let $I_{a,b}$ be a legal instance, i.e., a $K_{p,q}$-minor-free graph $G$ with a specific identity-assignment to its nodes, defined as follows. Let $d = \lfloor \frac{n}{2q} \rfloor$. The graph $G$ is defined as two disjoint paths, one of length $n_A$, and  another of length $n_B$. In $I_{a,b}$, the nodes in the former path are given IDs in $a$, in increasing order, while the nodes in the latter path are given IDs in $b$, also in increasing order. Also, for every $j \in \{1, \dots, q\}$,  there is an edge in $G$ between the node with ID $a[jd]$ and the node with ID $b[jd]$. An illustration of this construction is displayed on Figure \ref{fig:biplow1}.  

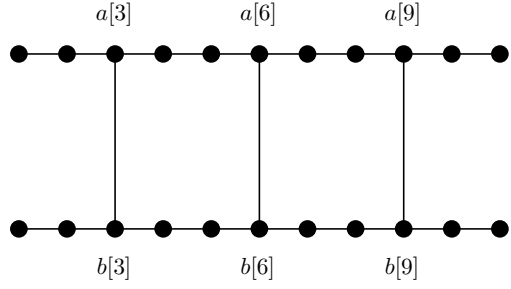
\begin{figure}[h]
\centering
\scalebox{0.8}{
\tikzset{every picture/.style={line width=0.75pt}} 

\begin{tikzpicture}[x=0.75pt,y=0.75pt,yscale=-1,xscale=1]

\draw  [fill={rgb, 255:red, 0; green, 0; blue, 0 }  ,fill opacity=1 ] (270,155) .. controls (270,152.24) and (272.24,150) .. (275,150) .. controls (277.76,150) and (280,152.24) .. (280,155) .. controls (280,157.76) and (277.76,160) .. (275,160) .. controls (272.24,160) and (270,157.76) .. (270,155) -- cycle ;
\draw  [fill={rgb, 255:red, 0; green, 0; blue, 0 }  ,fill opacity=1 ] (240,155) .. controls (240,152.24) and (242.24,150) .. (245,150) .. controls (247.76,150) and (250,152.24) .. (250,155) .. controls (250,157.76) and (247.76,160) .. (245,160) .. controls (242.24,160) and (240,157.76) .. (240,155) -- cycle ;
\draw  [fill={rgb, 255:red, 0; green, 0; blue, 0 }  ,fill opacity=1 ] (300,155) .. controls (300,152.24) and (302.24,150) .. (305,150) .. controls (307.76,150) and (310,152.24) .. (310,155) .. controls (310,157.76) and (307.76,160) .. (305,160) .. controls (302.24,160) and (300,157.76) .. (300,155) -- cycle ;
\draw  [fill={rgb, 255:red, 0; green, 0; blue, 0 }  ,fill opacity=1 ] (330,155) .. controls (330,152.24) and (332.24,150) .. (335,150) .. controls (337.76,150) and (340,152.24) .. (340,155) .. controls (340,157.76) and (337.76,160) .. (335,160) .. controls (332.24,160) and (330,157.76) .. (330,155) -- cycle ;
\draw  [fill={rgb, 255:red, 0; green, 0; blue, 0 }  ,fill opacity=1 ] (360,155) .. controls (360,152.24) and (362.24,150) .. (365,150) .. controls (367.76,150) and (370,152.24) .. (370,155) .. controls (370,157.76) and (367.76,160) .. (365,160) .. controls (362.24,160) and (360,157.76) .. (360,155) -- cycle ;
\draw  [fill={rgb, 255:red, 0; green, 0; blue, 0 }  ,fill opacity=1 ] (390,155) .. controls (390,152.24) and (392.24,150) .. (395,150) .. controls (397.76,150) and (400,152.24) .. (400,155) .. controls (400,157.76) and (397.76,160) .. (395,160) .. controls (392.24,160) and (390,157.76) .. (390,155) -- cycle ;
\draw  [fill={rgb, 255:red, 0; green, 0; blue, 0 }  ,fill opacity=1 ] (420,155) .. controls (420,152.24) and (422.24,150) .. (425,150) .. controls (427.76,150) and (430,152.24) .. (430,155) .. controls (430,157.76) and (427.76,160) .. (425,160) .. controls (422.24,160) and (420,157.76) .. (420,155) -- cycle ;
\draw  [fill={rgb, 255:red, 0; green, 0; blue, 0 }  ,fill opacity=1 ] (450,155) .. controls (450,152.24) and (452.24,150) .. (455,150) .. controls (457.76,150) and (460,152.24) .. (460,155) .. controls (460,157.76) and (457.76,160) .. (455,160) .. controls (452.24,160) and (450,157.76) .. (450,155) -- cycle ;
\draw  [fill={rgb, 255:red, 0; green, 0; blue, 0 }  ,fill opacity=1 ] (480,155) .. controls (480,152.24) and (482.24,150) .. (485,150) .. controls (487.76,150) and (490,152.24) .. (490,155) .. controls (490,157.76) and (487.76,160) .. (485,160) .. controls (482.24,160) and (480,157.76) .. (480,155) -- cycle ;
\draw  [fill={rgb, 255:red, 0; green, 0; blue, 0 }  ,fill opacity=1 ] (270,265) .. controls (270,262.24) and (272.24,260) .. (275,260) .. controls (277.76,260) and (280,262.24) .. (280,265) .. controls (280,267.76) and (277.76,270) .. (275,270) .. controls (272.24,270) and (270,267.76) .. (270,265) -- cycle ;
\draw  [fill={rgb, 255:red, 0; green, 0; blue, 0 }  ,fill opacity=1 ] (240,265) .. controls (240,262.24) and (242.24,260) .. (245,260) .. controls (247.76,260) and (250,262.24) .. (250,265) .. controls (250,267.76) and (247.76,270) .. (245,270) .. controls (242.24,270) and (240,267.76) .. (240,265) -- cycle ;
\draw  [fill={rgb, 255:red, 0; green, 0; blue, 0 }  ,fill opacity=1 ] (300,265) .. controls (300,262.24) and (302.24,260) .. (305,260) .. controls (307.76,260) and (310,262.24) .. (310,265) .. controls (310,267.76) and (307.76,270) .. (305,270) .. controls (302.24,270) and (300,267.76) .. (300,265) -- cycle ;
\draw  [fill={rgb, 255:red, 0; green, 0; blue, 0 }  ,fill opacity=1 ] (330,265) .. controls (330,262.24) and (332.24,260) .. (335,260) .. controls (337.76,260) and (340,262.24) .. (340,265) .. controls (340,267.76) and (337.76,270) .. (335,270) .. controls (332.24,270) and (330,267.76) .. (330,265) -- cycle ;
\draw  [fill={rgb, 255:red, 0; green, 0; blue, 0 }  ,fill opacity=1 ] (360,265) .. controls (360,262.24) and (362.24,260) .. (365,260) .. controls (367.76,260) and (370,262.24) .. (370,265) .. controls (370,267.76) and (367.76,270) .. (365,270) .. controls (362.24,270) and (360,267.76) .. (360,265) -- cycle ;
\draw  [fill={rgb, 255:red, 0; green, 0; blue, 0 }  ,fill opacity=1 ] (390,265) .. controls (390,262.24) and (392.24,260) .. (395,260) .. controls (397.76,260) and (400,262.24) .. (400,265) .. controls (400,267.76) and (397.76,270) .. (395,270) .. controls (392.24,270) and (390,267.76) .. (390,265) -- cycle ;
\draw  [fill={rgb, 255:red, 0; green, 0; blue, 0 }  ,fill opacity=1 ] (420,265) .. controls (420,262.24) and (422.24,260) .. (425,260) .. controls (427.76,260) and (430,262.24) .. (430,265) .. controls (430,267.76) and (427.76,270) .. (425,270) .. controls (422.24,270) and (420,267.76) .. (420,265) -- cycle ;
\draw  [fill={rgb, 255:red, 0; green, 0; blue, 0 }  ,fill opacity=1 ] (450,265) .. controls (450,262.24) and (452.24,260) .. (455,260) .. controls (457.76,260) and (460,262.24) .. (460,265) .. controls (460,267.76) and (457.76,270) .. (455,270) .. controls (452.24,270) and (450,267.76) .. (450,265) -- cycle ;
\draw  [fill={rgb, 255:red, 0; green, 0; blue, 0 }  ,fill opacity=1 ] (480,265) .. controls (480,262.24) and (482.24,260) .. (485,260) .. controls (487.76,260) and (490,262.24) .. (490,265) .. controls (490,267.76) and (487.76,270) .. (485,270) .. controls (482.24,270) and (480,267.76) .. (480,265) -- cycle ;
\draw    (245,155) -- (540,155) ;

\draw    (305,155) -- (305,265) ;

\draw    (395,150) -- (395,260) ;

\draw    (485,155) -- (485,265) ;

\draw    (245,265) -- (545,265) ;

\draw  [fill={rgb, 255:red, 0; green, 0; blue, 0 }  ,fill opacity=1 ] (510,155) .. controls (510,152.24) and (512.24,150) .. (515,150) .. controls (517.76,150) and (520,152.24) .. (520,155) .. controls (520,157.76) and (517.76,160) .. (515,160) .. controls (512.24,160) and (510,157.76) .. (510,155) -- cycle ;
\draw  [fill={rgb, 255:red, 0; green, 0; blue, 0 }  ,fill opacity=1 ] (540,155) .. controls (540,152.24) and (542.24,150) .. (545,150) .. controls (547.76,150) and (550,152.24) .. (550,155) .. controls (550,157.76) and (547.76,160) .. (545,160) .. controls (542.24,160) and (540,157.76) .. (540,155) -- cycle ;
\draw  [fill={rgb, 255:red, 0; green, 0; blue, 0 }  ,fill opacity=1 ] (540,265) .. controls (540,262.24) and (542.24,260) .. (545,260) .. controls (547.76,260) and (550,262.24) .. (550,265) .. controls (550,267.76) and (547.76,270) .. (545,270) .. controls (542.24,270) and (540,267.76) .. (540,265) -- cycle ;
\draw  [fill={rgb, 255:red, 0; green, 0; blue, 0 }  ,fill opacity=1 ] (510,265) .. controls (510,262.24) and (512.24,260) .. (515,260) .. controls (517.76,260) and (520,262.24) .. (520,265) .. controls (520,267.76) and (517.76,270) .. (515,270) .. controls (512.24,270) and (510,267.76) .. (510,265) -- cycle ;

\draw (304.5,130) node    {$a[ 3]$};
\draw (394.5,130) node    {$a[ 6]$};
\draw (484.5,130) node    {$a[ 9]$};
\draw (304.5,290) node    {$b[ 3]$};
\draw (394,290) node    {$b[ 6]$};
\draw (484,290) node    {$b[ 9]$};

\end{tikzpicture}
}
\vspace{-2ex}
\caption{An instance $I_{a,b}$ for $n=22$, $p=q=3$, and $d = 3$.}
\label{fig:biplow1}
\end{figure}

Observe that $G$ is outerplanar, and thus $G \in \Forb(\{ K_{2,3} \})$, from which it follows that  $G \in \Forb(\{ K_{p,q} \})$ for every $p \geq 2$ and $q \geq 3$.

For each node $v\in V(G)$, let $c_{a,b}(v)$ be the certificate provided to $v$ in instance $I_{a,b}$ leading all nodes to accept, and define 
\[
c(a,b) = \big(c_{a,b}(a[d]),c_{a,b}(a[2d]),\dots, c_{a,b}(a[qd]), c_{a,b}(b[d]),c_{a,b}(b[2d]),\dots, c_{a,b}(a[qd])\big).
\]
Since all certificates are of size $o(\log n)$, $c(a,b)$ forms a word on $o(\log n)$ bits. Now let $K_{n,n} = (A\cup B, E)$ be the complete bipartite graph with $E = \{\{a,b\}: a\in A, b\in B\}$, where every edge $\{a,b\}\in E$ is colored by $c(a,b)$. Pick~$n$ sufficiently large such that the number of bits of $c(a,b)$ is smaller than $\frac{\log(n)}{2q}$. Then, there exists a monochromatic subset $F$ of edges of $K_{n,n}$ with 
\[
|F| > |E|/n^{\frac{1}{2q}} = n^{2-\frac{1}{2q}}.
\] 
It is known~\cite{furedi1996upper} that every graph with $n^{2-1/q} + o(n^{2-1/q})$ edges contains $K_{q,q}$ as subgraph. It follows that, for sufficiently large~$n$, the graph $H=(A\cup B, F)$ contains $K_{q,q}$ as subgraph. Up to reindexing the elements in the sets~$A$ and~$B$, the vertices of this subgraph are $a_1, \dots, a_q$ and  $b_1, \dots, b_q$. Since all edges in~$H$ have  the same color, 
\[
c(a_{i}, b_{j}) = c(a_{i'}, b_{j'}), 
\]
for every $(i, i', j, j') \in \{1, \dots, q\}^4$. We now create a new instance $J$, that roughly consists in gluing together the instances $I_{a_i,b_j}$ with $1\leq i,j\leq q$, for constructing an illegal instance. Let us consider $q$ disjoint copies $P_1,\dots,P_q$ (resp., $Q_1,\dots,Q_q$) of a path~$P$ of length $n_A$ (resp., of a path~$Q$ of length $n_B$), such that, in $P_i$ (resp., $Q_i$), $i=1,\dots,q$, the nodes are given the IDs in $a_i$ (resp., $b_i$) in increasing order. For every $i,j \in \{1, \dots, q\}$, the nodes $a_i[jd]$ and $b_{i+j}[jd]$, where by $i+j$ is taken modulo $q$ whenever greater than~$q$, are connected by an edge. Figure \ref{fig:biplow2} displays an example of the construction.  

\begin{figure}[h]
\centering
\scalebox{0.85}{
\input{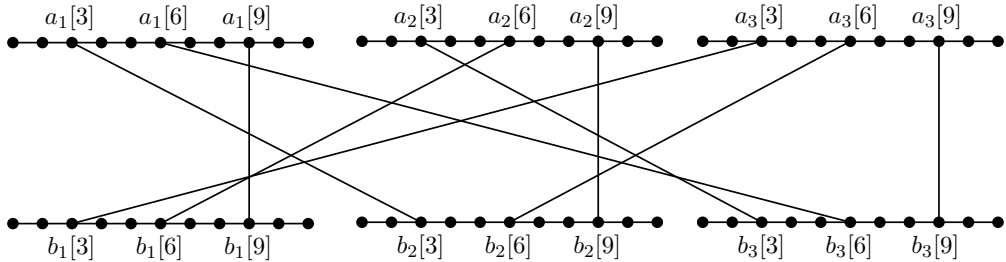}
}
\vspace{-2ex}
\caption{Example of the instance $J$ for $p=q=3$, and $d = 3$.}
\label{fig:biplow2}
\end{figure}

Observe that, by construction,  if each of the $q$ paths $P_1,\dots,P_q$, and each of the  $q$ paths $Q_1,\dots,Q_q$ is contracted into a single node, then the resulting graph is  $K_{q,q}$. In other words, the underlying network in $J$ contains $K_{q,q}$ as a minor, and therefore also $K_{p,q}$ as a minor. It follows that the instance $J$ is illegal. 

Let us consider the following certificate assignment $c_J$ to the nodes in instance~$J$.  For every $i\in\{1,\dots,q\}$, every node $u$ of the path $P_i$ (resp., $Q_i$) is given the certificate that it would receive in $I_{a_i, b_i}$. It follows from this assignment that, for every node $u$ whose ID is different from $a_i[jd]$ and $b_i[jd])$, for all $ i,j \in \{1, \dots, q\}$, then the closed neighborhood of~$u$ in $J$ with certificate assignment~$c_J$ is identical to the closed neighborhood of the same node~$u$ in $I_{a_i,b_i}$ with certificates assignment $c_{a_i, b_i}$. It follows that such a node~$u$ accepts in $J$. Let $u$ be a node of $P_i$ with ID $a_i[jd])$ for some $i,j\in\{1,\dots,q\}$.  The closed neighborhood of~$u$ in $J$ with certificate assignment~$c_J$ is identical to the view of the same node~$u$ in $I_{a_i,b_{i+j}}$ with certificates assignment $c_{a_i, b_{i+j}}$. It follows that such a node~$u$ accepts in $J$ as well. The case of a node of $Q_i$ with ID $b_i[jd])$ for some $i,j\in\{1,\dots,q\}$ is identical, i.e., such a node accepts in $J$ too. Therefore, with certificate assignment $c_J$, all nodes accept in $J$, and thus soundness is not satisfied. This contradicts the existence of a locally checkable proof for $\Forb(\{ K_{p,q} \})$ with $o(\log n)$ bits, and completes the proof of Lemma~\ref{lem:Kpq}. 
\end{proof}

Again, the proof can be adapted to a larger verification radius $t$. The instances are formed by paths with identifiers only in $a$, or only in $b$, and edges having an endpoint in $a$, and an endpoint in $b$. We replace each edge of the former type (identifiers only in $a$, or only in $b$) by a path of length~$t$. The same proof as for Lemma~\ref{lem:Kpq} applies after this change, in particular the legal instances are still outerplanar.

\subsection{Proof of Theorem~\ref{theo:lwb}}
\label{subsec:lower-bound-proof}

Let $\mathcal{F} = \{K_k:k\geq 3\}\cup\{K_{p,q}: p,q\geq 2\}$, and $\cal{H}$ be a finite set of graphs in $\mathcal{F}$. 
If $\cal{H}$ contains a unique graph in $\mathcal{F}$, the fact that there are no locally checkable proofs for $\Forb(\mathcal{H})$ using certificates on $o(\log n)$ bits has been established in Lemma~\ref{lem:Kk} and Lemma~\ref{lem:Kpq}.  
We are left with the case of $|\mathcal{H}|\geq 2$. 

Let us first make a few remarks about the family $\mathcal{H}$. 
If $\mathcal{H}$ contains two cliques $K_k$ and $K_{k'}$ with $k'>k$, then we can remove $K_{k'}$, because $K_k$-minor-freeness implies $K_{k'}$-minor-freeness. Thus $\mathcal{H}$ contains at most one clique.
Similarly, if $\mathcal{H}$ contains two complete bipartite graphs $K_{p,q}$ and $K_{p',q'}$, with $p' \geq p$ and $q' \geq q$, then we can safely remove $K_{p',q'}$ from the family. 
Also, if $\mathcal{H}$ contains $K_3$, then, $\Forb(\mathcal{H})$ is the class of trees, for which it is known that $o(\log n)$ bits are not sufficient for certification. Similarly, if $\mathcal{H}$ contains $K_{2,2}$ (and no $K_3$) then $\Forb(\mathcal{H})$ is the class of the graphs that do not contain cycles larger than 4, and, as noted in the proof of Lemma~\ref{lem:Kpq}, this is also a case where $o(\log n)$ is not enough.

We are then left with the case of a family $\mathcal{H}$ included in $\{K_k:k\geq 4\}\cup\{K_{p,q}: p\geq 2,q\geq 3\}$, containing exactly one clique $K_k$, $k\geq 4$, and one or many complete bipartite graphs $K_{p,q}$, $p\geq 2,q\geq 3$. Consequently, the class of graphs $\Forb(\cal{H})$ contains the class $\Forb(\{K_4,K_{2,3}\})$. The latter is precisely the class of outerplanar graphs. 
Recall that in the construction of proof of Lemma~\ref{lem:Kpq} the legal instances are outerplanar, therefore these instances must be in $\Forb(\cal{H})$. The construction of the proof of Lemma~\ref{lem:Kpq} shows that if a scheme with $o(\log n)$ bits accepts all these instances, it also accepts an instance with $K_{p,q}$ as a minor, which is a contradiction. The result follows. \qed

\section{Conclusion and further work}

This paper provides a proof-labeling scheme for planarity, using certificates of optimal size $\Theta(\log n)$ bits, hence improving the previously known \dMAM\/ interactive protocol~\cite{NaorPY20}, by decreasing the number of interactions between the prover and the verifier, and avoiding the use of randomization. 

This work could find extensions in many directions, fitting with the aforementioned interest of the distributed computing community for classes of sparse graphs beyond planar graphs (see~\cite{Feuilloley20}). For instance, the approach in this paper appears to be generalizable to the design of a proof-labeling scheme for the class of graphs with bounded genus $g>1$, still using certificates on $O(\log n)$ bits, by cutting along so-called \emph{nooses}  until the resulting surface is planar. Certifying the correctness of such cuts, and the consistency of the embedding after these cuts requires to fix many technical details, but appear doable with  certificates on $O(\log n)$ bits. 

Perhaps more interesting is to extend our result to the certification of graph classes defined by a finite set of forbidden minors.  The technique in this paper also enables to certify the class of outerplanar graphs with certificates on $O(\log n)$ bits, but the general case of classes defined by forbidden minors is open. A standard tool to consider in general for graphs excluding a fixed minor, that is, to certify $G\in\Forb(H)$ for a fixed graph~$H$, is the characterization of the  topological embeddings for these graphs provided by Robertson and Seymour's Theorem. However, certifying the correctness of such embeddings seems challenging using certificates of small size. In fact, the minor-excluding graphs have a linear number of edges, and therefore they admit proof-labeling schemes using certificates on $O(n\log n)$ bits~\cite{GoosS16,KormanKP10}. Even the design of proof-labeling schemes for these graph classes with certificates on $O(n^\epsilon)$ bits, $\epsilon<1$, appears challenging.

\bibliographystyle{plain}
\bibliography{pls-planar}

\end{document}